\documentclass[twoside]{article} %

\usepackage[ruled,vlined,linesnumbered]{algorithm2e}
\SetKwInput{KwInput}{Input}
\SetKwInput{KwOutput}{Output}

\usepackage[accepted]{aistats2021}

\usepackage[round]{natbib}

\usepackage[utf8]{inputenc}       %
\usepackage[T1]{fontenc}          %
\usepackage[dvipsnames]{xcolor}
\usepackage[colorlinks,citecolor=ForestGreen]{hyperref} %
\usepackage{url}                  %
\usepackage{dsfont}
\usepackage{booktabs}             %
\usepackage{amsfonts}             %
\usepackage{nicefrac}             %
\usepackage{microtype}            %
\usepackage{amsmath}
\usepackage{physics} %
\usepackage{graphicx}
\usepackage{floatrow}
\usepackage{subfig}
\usepackage{tabularx}
\usepackage{theoremref}
\usepackage{tablefootnote}
\usepackage{bm}

\usepackage{tikz}
\usepackage{pgfplots}
\pgfplotsset{compat=newest}
\usepgfplotslibrary{groupplots}
\usepgfplotslibrary{polar}
\usepgfplotslibrary{smithchart}
\usepgfplotslibrary{statistics}
\usepgfplotslibrary{dateplot}
\usepgfplotslibrary{fillbetween}
\usetikzlibrary{backgrounds}
\usetikzlibrary{shapes,arrows}
\usetikzlibrary{arrows,automata}
\usetikzlibrary{external}
\usepackage{makecell}
\usepackage{multirow}
\definecolor{forestgreen}{rgb}{0.13, 0.55, 0.13}
\usepackage{floatrow}
\newfloatcommand{capbtabbox}{table}[][\FBwidth]

\usepackage{amsthm}
\theoremstyle{plain}
\newtheorem{theorem}{Theorem}[section]
\newtheorem{lemma}{Lemma}[section]

\newtheorem{proposition}{Proposition}[section]

\newtheorem{assumption}{Assumption}
\newtheorem{condition}{Condition}

\theoremstyle{definition}

\usepackage{thmtools}

\declaretheoremstyle[%
  spaceabove=0pt,%
  spacebelow=0pt,%
  headfont=\normalfont\itshape,%
  postheadspace=1em,%
  qed=\qedsymbol%
]{mystyle} 
\declaretheorem[name={Proof},style=mystyle,unnumbered,
]{prf}

\usepackage{bbm}
\usepackage{amssymb}
\usepackage{siunitx}

\newcommand{\assumptref}[1]{Assumption~\ref{#1}}
\newcommand{\condref}[1]{Condition~\ref{#1}}
\newcommand{\figref}[1]{Figure~\ref{#1}}
\newcommand{\algref}[1]{Algorithm~\ref{#1}}
\newcommand{\secref}[1]{Section~\ref{#1}}
\newcommand{\thmref}[1]{Theorem~\ref{#1}}
\newcommand{\lemref}[1]{Lemma~\ref{#1}}
\newcommand{\propref}[1]{Proposition~\ref{#1}}
\newcommand{\appref}[1]{Appendix~\ref{#1}}

\newcommand*{\wc}{{\mkern 2mu\cdot\mkern 2mu}} %

\newcommand{\PkerHMC}[2]{\mathrm{pr}_{#1, #2}}
\newcommand{\PkerCHMC}[3]{\mathrm{pr}_{#1, #2}^{#3}}
\newcommand{\PkerCMHMC}[4]{\mathrm{pr}_{#1, #2, #3}^{#4}}

\newcommand{\stepsize}{\varepsilon}

\newcommand{\set}[1]{\left\{ #1 \right\}}

\newcommand{\flow}[1]{\Phi_{#1}}
\newcommand{\flowq}[1]{\Phi_{#1}^{\circ}}
\newcommand{\flowp}[1]{\Phi_{#1}^{ * }}

\newcommand{\flowdisc}[2]{\hat{\Phi}_{#1, #2}}
\newcommand{\flowqdisc}[2]{\hat{\Phi}_{#1, #2}^{\circ}}
\newcommand{\flowpdisc}[2]{\hat{\Phi}_{#1, #2}^{ * }}

\newcommand{\levelset}[2]{L_{#1}(#2)}

\DeclareMathOperator*{\argmin}{arg\,min}

\DeclareMathOperator{\MKer}{\mathcal{K}} %
\DeclareMathOperator{\CMKer}{\bar{\MKer}} %
\DeclareMathOperator{\KerHMC}{\mathcal{K}} %
\DeclareMathOperator{\KerCHMC}{\bar{\KerHMC}} %
\DeclareMathOperator{\Energy}{\mathcal{E}}
\DeclareMathOperator{\Leb}{Leb}

\DeclareMathOperator{\TV}{D_{TV}}

\DeclareMathOperator{\Cat}{\mathcal{C}at}
\DeclareMathOperator{\Unif}{\mathcal{U}}
\DeclareMathOperator{\Normal}{\mathcal{N}}
\DeclareMathOperator{\Normald}{\mathcal{N}_d}
\DeclareMathOperator{\Exp}{\mathcal{E}xp}

\newcommand{\mtime}{\tau}
\newcommand{\traj}{\mathbf{t}}

\newcommand{\NN}{\mathbb{N}}

\newcommand{\RR}{\mathbb{R}}

\newcommand{\ZZ}{\mathbb{Z}}
\newcommand{\PP}{\mathbb{P}}

\newcommand{\vmu}{\bm{\mu}}
\newcommand{\vnu}{\bm{\nu}}
\newcommand{\vv}{\mathbf{v}}
\newcommand{\vx}{\mathbf{x}}
\newcommand{\vy}{\mathbf{y}}

\newcommand{\calB}{\mathcal{B}}

\newcommand{\calO}{\mathcal{O}}
\DeclareMathOperator{\ind}{\mathds{1}}
\newcommand{\rmd}{\mathrm{d}}

\newcommand{\E}[2][]{\mathbb{E}_{#1} \left[ #2 \right]}
\newcommand{\Prob}[2][]{\mathbb{P}_{#1} \left(#2\right)}

\usepackage{pifont}

\usepackage[normalem]{ulem}

\begin{document}

\twocolumn[

\aistatstitle{Couplings for Multinomial Hamiltonian Monte Carlo}

\aistatsauthor{ Kai Xu$^\ast$ \And Tor Erlend Fjelde$^\ast$ \And Charles Sutton \And Hong Ge }

\aistatsaddress{ University of Edinburgh \And  University of Cambridge \And University of Edinburgh \And University of Cambridge } 

]

\begin{abstract}
Hamiltonian Monte Carlo (HMC) is a popular sampling method in Bayesian inference. Recently, Heng \& Jacob (2019) studied Metropolis HMC with couplings for unbiased Monte Carlo estimation, establishing a generic parallelizable scheme for HMC. However, in practice a different HMC method, multinomial HMC, is considered as the go-to method, e.g. as part of the no-U-turn sampler. In multinomial HMC, proposed states are not limited to end-points as in Metropolis HMC; instead points along the entire trajectory can be proposed. In this paper, we establish couplings for multinomial HMC, based on optimal transport for multinomial sampling in its transition. We prove an upper bound for the meeting time -- the time it takes for the coupled chains to meet -- based on the notion of local contractivity. We evaluate our methods using three targets: $1{,}000$ dimensional Gaussians, logistic regression and log-Gaussian Cox point processes. Compared to Heng \& Jacob (2019), coupled multinomial HMC generally attains a smaller meeting time, and is more robust to choices of step sizes and trajectory lengths, which allows re-use of existing adaptation methods for HMC. These improvements together paves the way for a wider and more practical use of coupled HMC methods.
\end{abstract}
\vspace{-1em}

\section{Introduction}
Markov chain Monte Carlo (MCMC) is a standard tool to draw samples 
from target distributions known up to a normalising constant \citep{metropolis1953equation,Geman1984-qq}.
Such samples are commonly used to estimate an integral of interest. 
Specifically, 
for a probability distribution $ \pi $ on $ \RR^d $ and 
a measurable function of interest $ h : \RR^d \mapsto \RR $, we want to estimate
\begin{equation}
H 
= \int \pi(x) h(x) \rmd x 
= \E[x \sim \pi]{h(x)}.
\label{eq:H}
\end{equation}
Approximating this integral $ H $ in \eqref{eq:H} is at the core of many statistics and machine learning problems.
For example in Bayesian inference, Monte Carlo samples are used to estimate some posterior predictive distribution, or perform model comparison \citep{gelman2013bayesian}.
Or in energy-based modelling, MCMC samples are to estimate gradients used to update model parameters \citep{teh2003energy,xie2016theory,qiu2019unbiased}.

Under the framework of MCMC, 
a Markov chain is simulated to obtain correlated samples from $\pi$, and then
Monte Carlo integration is used to estimate $H$ using these samples.
However, such estimators are unbiased only when the underlying Markov chain has converged to the equilibrium, which is challenging to verify in practice.
Therefore MCMC with couplings has attracted research attention recently thanks to its ability to debias Monte Carlo estimators \citep{jacob_unbiased_2019}.
In particular, 
\citet{heng_unbiased_2019} focused on the Metropolis-Hastings (MH) adjusted HMC variant, 
which proposes the end-point of a simulated Hamiltonian trajectory as the new state, followed by an MH correction step. 
We refer to this HMC variant as \textit{coupled Metropolis HMC}.
\citet{heng_unbiased_2019} noticed that coupled Metropolis HMC is sensitive to the choice of HMC parameters such as integrator step sizes and Hamiltonian trajectory lengths.
More specifically, 
parameters (e.g. trajectory lengths) optimal for sampling efficiency (e.g. effective sample size) can require a large number of HMC iterations to achieve meeting; on the other hand, optimal parameters for coupling can lead to poor mixing \citep{heng_unbiased_2019}.

Building upon the recent work of \citet{heng_unbiased_2019}, we propose two novel couplings based on a different, more robust implementation of HMC.
We refer to our methods as \emph{coupled multinomial HMC} and demonstrate several advantages of these methods. 
First, coupled multinomial HMC meets faster in general. 
Intuitively, like all MH algorithms, 
the previous coupled HMC method can only propose a point from the initial or the last integration step, 
which leads to two drawbacks for couplings: (i) it may well be that intermediate points are the closest between two chains and (ii) rejection rates of proposals are quite sensitive to step sizes of Hamiltonian dynamics solvers.
Multinomial coupling allows coupled chains to accept intermediate points that are potentially closer together, so they meet quicker.
We therefore design couplings to minimize the expected distance between coupled chains within each transition to encourage faster meeting.
Second, coupled multinomial HMC is less sensitive to Hamiltonian integration step sizes. For Metropolis HMC, 
a small enough step size has to be used to ensure a large enough acceptance probability in the MH adjustment step. 
However, for multinomial HMC, 
intermediate points can be proposed 
even though the end-points would have
been rejected in Metropolis HMC.
We argue that this robustness is crucial for practical use of coupled HMC algorithms.
Thirdly, we prove that
the meeting time of coupled multinomial HMC decays geometrically,
which is a sufficient condition to use the unbiased estimator from \citet{jacob_unbiased_2019}.
Finally, we perform extensive simulations to verify the improved meeting and robustness of our proposed method.

\section{Background}
\subsection{Unbiased MCMC with couplings}
\label{sec:unbiased_mcmc}

For two distributions $ p $ and $ q $,
we denote $ \Gamma(p, q) $ as their \textit{couplings},
i.e. for any $ \gamma \in \Gamma(p, q) $,
the marginals of $ \gamma $ are $p$ and $q$.
For a $\pi$-invariant Markov kernel $\MKer$ defined on $\left( \RR^d, \calB(\RR^d) \right)$,
its coupled kernel $\CMKer$, 
defined on $\left( \RR^d \times \RR^d, \calB(\RR^d) \times \calB(\RR^d) \right)$,
by construction has $\MKer$ as its marginals,
where $ \calB $ denotes the Borel $\sigma$-algebra.
Additionally, given an initial distribution $ \pi_0 $ and some $ \bar{\pi}_0 \in \Gamma(\pi_0, \pi_0)$,
a pair of coupled chains $X=(X_n)_{n \geq 0}$, $Y=(Y_n)_{n \geq 0}$ that share the same equilibrium distribution $\pi$ can be simulated by \algref{alg:coupled_chains} \citep{jacob_unbiased_2019} until meeting at iteration $\mtime := \inf \{n \ge 1 : X_n = Y_{n - 1}\}$ (the \textit{meeting time}).
\begin{algorithm}[t]
Sample $(X_0, Y_0) \sim \bar{\pi}_0$ ($\bar{\pi}_0$ is a coupling of $\pi_0$)\;
Sample $X_1 \sim \MKer(X_0, \wc)$\;
Set $N=1$\;
\While{$X_N \neq Y_{N-1}$}{
    Sample $(X_{N+1}, Y_N) \sim \CMKer((X_N, Y_{N+1}), \wc)$\;
    Set $N=N+1$\;
}
Set $\mtime = N$ and output $\{(X_n)_{n=0}^\mtime, (Y_n)_{n=0}^{\mtime-1}\}$\;
\caption{Sample a pair of coupled chains}
\label{alg:coupled_chains}
\end{algorithm}
The main design choice in this algorithm is the construction of $\CMKer$.
\citet{jacob_unbiased_2019} established that if $\CMKer$ satisfies 
certain conditions,
a pair of coupled chains $X, Y$ from Algorithm~\ref{alg:coupled_chains} can be used to obtain unbiased estimates of~\eqref{eq:H} with finite variance and finite computation cost as\looseness=-1
\begin{equation}
H_k(X, Y) = h(X_k) + {\sum}_{n=k+1}^{\mtime-1}(h(X_n) - h(Y_{n-1}))
\label{eq:H_k}
\end{equation}
where $ k \in \NN $ is a parameter to choose.
The first term in~\eqref{eq:H_k} is a standard, single-sample MCMC estimate and the second term 
can be seen as a debiasing term to correct the bias introduced by non-converged chains.
This estimator is built on the pioneering works from \citet{glynn2014exact}, derived using telescoping sums.
In practice, we use a time-averaged version of \eqref{eq:H_k}, which is still unbiased but with lower variance, e.g.~in Section~\ref{sec:exp}.

\subsection{Hamiltonian Monte Carlo}
In an HMC kernel, new states are proposed by simulating Hamiltonian dynamics \citep{neal_mcmc_2012}.
For a Hamiltonian system with a position variable $ q \in \RR^d $ and 
a momentum variable $ p \in \RR^d $, the trajectory $ \traj := \left( q(t), p(t) \right)_{t \in \RR_+} $ can be described by the following ordinary differential equations
\begin{equation}
\begin{aligned}
&\dv{q}{t} = + \nabla_p \Energy \left( q(t), p(t) \right),\\
&\dv{p}{t} = - \nabla_q \Energy \left( q(t), p(t) \right) = - \nabla U \left( q(t) \right)
\end{aligned}
\label{eq:hamiltonian-system}
\end{equation}
where the \textit{potential} $ U: \RR^d \mapsto \RR_+ $ is chosen s.t.~
the target $ \pi(q) \propto \exp \left( - U(q) \right) $,
the \textit{kinetic} term $ K: \RR^d \mapsto \RR_+ $ is 
assumed to have a form of $ K(p) = \frac{1}{2} p^\top M p $, where $ M $ is the mass matrix, and
the \textit{Hamiltonian} is defined as $ \Energy(q, p) := U(q) + K(p) $.\footnote{
Unless otherwise specified we let $M^{-1} = I_d$ throughout, though we note that $M$ can be chosen using existing adaption methods, e.g. \citep{carpenter2017stan}, or as in Riemannian HMC \citep{girolami2011riemann}.}
The extended target $\bar{\pi}$ for \textit{phase points} $ z := (q, p)$ on the \textit{phase space} $ \RR^d \times \RR^d $ is then defined as having density $ \propto \exp \left( - \Energy(q, p) \right) $.

To describe the dynamics more succinctly, 
we consider the \textit{flow} map $\flow{t}(q_0, p_0) = \left( q(t), p(t) \right)$ for~\eqref{eq:hamiltonian-system} initialized at $ (q_0, p_0) := \big( q(0), p(0) \big) \in \mathbb{R}^d \times \mathbb{R}^d $.
Following \citet{heng_unbiased_2019}, 
we write $ \flowq{t}(q_0, p_o) = q(t) $ and $ \flowp{t}(q_0, p_0) = p(t) $ for the flow projected onto its position and momentum spaces, respectively. 
The flow map $ \flow{t} $ is in general not available in closed form and requires discretization in time via numerical integrators as approximations. 
A standard choice for HMC is the \emph{leapfrog integrator} that, given an initial phase point $ (q_0, p_0) $, iterates:
\begin{equation*}
\begin{split}
    p_{\ell + 1 / 2} &:= p_{\ell} - \frac{\stepsize{}}{2} \nabla U(q_{\ell}) \\
    q_{t + 1} &:= q_{\ell} + \stepsize{} p_{\ell + 1 / 2} \\
    p_{t + 1} &:= p_{\ell + 1/2} - \frac{\stepsize{}}{2} \nabla U(q_{\ell + 1})
\end{split}
\end{equation*}
for $\ell = 0, \dots, L - 1$ with a step size $ \stepsize{} > 0 $ and leapfrog steps $ L \in \NN $.
We denote $ \flowdisc{\stepsize{}}{\ell}(q_0, p_0) := (q_{\ell}, p_{\ell})$ as the numerical flow map approximated by a leapfrog integrator with a step size $ \stepsize{} $ for $ \ell $ steps, and
similarly $\flowqdisc{\stepsize{}}{\ell}$ and $\flowpdisc{\stepsize{}}{\ell}$ for projected maps onto position and momentum.\looseness=-1%

\paragraph{Metropolis HMC}

One can design an MCMC kernel by proposing the \textit{end-point} of a Hamiltonian trajectory in~\eqref{eq:hamiltonian-system}.
In practice, a discretized trajectory is obtained by leapfrog integration.
Due to numerical errors in the simulation,
in order to ensure the kernel $ \pi $-invariant,
the proposal needs to be adjusted by a Metropolis-Hasting step \citep{metropolis1953equation,neal_mcmc_2012}.
Denoting $\Normald$ as the $d$-dimensional standard Gaussian, the kernel $Q \sim \MKer_{\stepsize{},L}^{\text{MH}}(Q_0, \wc)$ for Metropolis HMC follows
\begin{equation}
\begin{aligned}
P_0 &\sim \Normald,\quad
(q_L, p_L) = \flowdisc{\stepsize{}}{L}(Q_0, P_0),\\
Q &=
\begin{cases}
    q_L &\text{ with prob. } \min \set{1, \exp ( \Delta_{\Energy} ) } \\
    Q_0   &\text{ otherwise}
\end{cases}
\end{aligned}
\label{eq:metropolis-hmc}
\end{equation}
where $ \Delta_{\Energy} := -\Energy(q_L, p_L) + \Energy(Q_0, P_0) $ is the energy difference between the origin and the proposal.

\paragraph{Multinomial HMC}

\citet{betancourt_conceptual_2018} describes a trajectory variant of HMC,
which we refer as \textit{multinomial} HMC and denote $ \MKer_{\stepsize{},L}^{\text{Mult}}(Q_0, \wc) $.
In multinomial HMC, 
all intermediate points of a numerical trajectory can be proposed as the next state:
\begin{equation}
P_0 \sim \Normald,
\traj \sim \PP_{\stepsize{}, L}(\cdot \mid Q_0, P_0),
(Q, P) \sim \PP(\cdot \mid \traj)
\label{eq:multi-hmc}
\end{equation}
where $\traj := [ (q_{-L_{\text{b}}}, p_{-L_{\text{b}}}), \dots, (Q_0, P_0), \dots, (q_{L_{\text{f}}}, p_{L_{\text{f}}}) ] $. 
The \textit{trajectory sampling} $ \traj \sim \PP_{\stepsize{}, L}(\cdot \mid Q_0, P_0)$ follows
\begin{equation}
\begin{aligned}
L_{\text{f}} &\sim \Unif(\{0, \dots, L\}),\quad
L_{\text{b}} = L - L_{\text{f}}, \\
(q_\ell, p_\ell)  &=
\begin{cases}
    \flowdisc{\stepsize{}}{\ell}(Q_0, +P_0)  &\text{ for } \ell = 1, \dots, L_{\text{f}} \\
    \flowdisc{\stepsize{}}{\ell}(Q_0, -P_0) &\text{ for } \ell = 1, \dots, L_{\text{b}}
\end{cases}
\end{aligned}
\label{eq:multi-traj}
\end{equation}
The \textit{intra-trajectory sampling} $Q, P \sim \PP(\cdot \mid \traj)$ follows a multinomial distribution \citep{betancourt_conceptual_2018} as
\begin{equation}
\PP\big((Q, P) = (q_\ell, p_\ell) \mid \traj \big) = \sigma\big((q_\ell, p_\ell), \traj\big)
\label{eq:multi-intra}
\end{equation}
where $\sigma(z, \traj) := \exp \left( -\Energy(z) \right) / \sum_{z' \in \traj} \exp \left( -\Energy(z') \right)$.

\subsection{Coupled MCMC kernels}

The coupled HMC kernel in \citep{heng_unbiased_2019} \textit{and} 
the coupled kernels proposed in this work can be unified through Algorithm~\ref{alg:coupled-hmc}.
\begin{algorithm}[t]
Sample $P_0 \sim \Normald $ \;%
Sample $L_\text{f} \sim \PP_{L_\text{f}}$ and set $L_\text{b} = L - L_\text{f}$ \;
\For{$c = 1, 2$}{
    Integrate using leapfrog to obtain $\traj^c = [\flowdisc{\stepsize{}}{-L_\text{b}}(Q_0^c, -P_0), \dots (Q_0^c, P_0), \dots \flowdisc{\stepsize{}}{L_\text{f}}(Q_0^c, P_0)] $
}
Sample next state indices $(i, j) \mid (\traj^1, \traj^2) \sim \bar{P}_\ell$\;
Set $(Q^1, P^1) = \traj^1_i, (Q^2, P^2) = \traj^2_j $\;
Output $(Q^1, Q^2)$\;
\caption{Coupled HMC kernels}
\label{alg:coupled-hmc}
\end{algorithm}
Specifically, letting (i) In Line 2, $\PP_{L_\text{f}}(L_\text{f}=L) = 1$ and (ii) In Line 5, $(i, j) \mid (\traj^1, \traj^2) \sim \bar{P}_\ell$ in \algref{alg:coupled-hmc}, we recover the $\CMKer_{\stepsize{}, L}^{\text{MH}}$ from \citet{heng_unbiased_2019}, where $\bar{P}_\ell$ follows the generative process
\begin{equation*}
\begin{aligned}
    u \sim &\Unif([0, 1]),\; \\
    i =
    \begin{cases}
        L &\text{if } u < \alpha^1 \\
        0 &\text{otherwise }
    \end{cases}\;
    \quad &\text{and} \quad
    j  =
    \begin{cases}
        L &\text{if } u < \alpha^2 \\
        0 &\text{otherwise }
    \end{cases}
\end{aligned}
\end{equation*}
where $\alpha^c = \exp \left( -\Energy(\traj_L^c) + \Energy(\traj_0^c)\right)$ for $c = 1, 2$.
This corresponds to using \textit{common random number} (CRN) in the MH correction steps in~\eqref{eq:metropolis-hmc}.
For coupled multinomial HMC kernels studied in this work, which we denote $\KerCHMC_{\stepsize{}, L}^{\gamma}$, we make different choices for Line 2 and Line 5 in \algref{alg:coupled-hmc}.
In short, 
Line 2 will be a coupled version of~\eqref{eq:multi-traj} and 
Line 5 will correspond to a coupling $\gamma$ of two multinomial distributions as ~\eqref{eq:multi-intra}.
We will discuss them in detail in Section~\ref{sec:mchmc}.

Although coupled HMC kernels can bring two chains within a small neighborhood of each other, the probability of \emph{exact} meeting is zero, thus failing to satisfy conditions for using~\eqref{eq:H_k}.
To alleviate this issue
\citet{heng_unbiased_2019} instead propose a mixture of coupled random-walk Metropolis-Hastings (RWMH) $\CMKer_\sigma$ and coupled HMC $\CMKer_{\stepsize{},L}$ to trigger ``exact meeting''.
The coupled RWMH kernel $\CMKer_\sigma$ with proposal variance $\sigma^2 I_d$ uses maximal coupling \citep{johnson1998coupling,jacob_unbiased_2019} to encourage two chains meet exactly when they are close.
For completeness, we provide it as Algorithm~\ref{alg:coupled-mh} in Appendix~\ref{app:background}.
The overall mixture kernel, denoted $\CMKer_{\stepsize{}, L, \sigma}$, is then defined as
\begin{equation}\label{eq:mixture_kernel}
\CMKer_{\stepsize{}, L, \sigma} \left(\bar{x}, \bar{A} \right) =
(1 - \alpha) \CMKer_{\stepsize{}, L} \left( \bar{x}, \bar{A}  \right) + \alpha \CMKer_{\sigma} \left( \bar{x}, \bar{A}  \right)
\end{equation}
for $\alpha \in (0, 1)$, $\bar{x} := (x, y) \in \RR^d \times \RR^d$ and $\bar{A} := (A, B) \in \calB(\RR^d) \times \calB(\RR^d)$.
That is, with probability $\alpha$ we use the coupled RWMH kernel and with probability $1 - \alpha$ we use the HMC kernel.
\citet{heng_unbiased_2019} proves that, under certain assumptions,
if the \textit{relaxed meeting time} $\tau_\delta := \inf \{ n \geq 0 : \norm{X_n - Y_{n-1}}  \leq \delta \} $ of the coupled HMC kernel $\CMKer_{\stepsize, L}$ has geometric tails for any $\delta > 0$, the chains meet exactly with non-zero probability under $\CMKer_{\stepsize{}, L, \sigma}$ for any $\alpha \in (0, 1)$,
warranting the use of~\eqref{eq:H_k}.

The key conditions that ensure the unbiasedness, finite variance and finite computation cost of~\eqref{eq:H_k} are (i) the coupled chains marginally converge to the target and (ii) the two chains meet sufficiently quickly and stay together after meeting; see \citet{jacob_unbiased_2019} for explicit definitions.
Suppose our proposed HMC kernels satisfy (i) by construction, to ensure the method satisfies (ii) it is sufficient to prove that the relaxed meeting time has geometric tails. 
This we establish in \secref{sec:theoretical}.\looseness=-1

\section{Optimal Transport Couplings for Multinomial HMC}
\label{sec:mchmc}

Recall that in order to use the multinomial HMC kernel $\MKer_{\stepsize{},L}^{\text{Mult}}$ in Algorithm~\ref{alg:coupled-hmc},
we need to specify how Line 2 and Line 5 are performed.
First, the number of leapfrog steps forward and backward sampled in Line 2 follows~\eqref{eq:multi-traj}, inheriting from multinomial HMC, and is shared between the two chains.
In other words, both chains simulate forward and backward for the same number of steps, making them ``aligned in-time''.
Second, Line 5 correspond to a coupling of the intra-trajectory multinomial sampling step in~\eqref{eq:multi-intra}.
To ensure that the marginal chains are equivalent to the original multinomial kernel, it is sufficient to sample $ (i,j) $ such that
the corresponding marginal \textit{categorical distributions} $\vmu$ and $\vnu$ 
of~\eqref{eq:multi-intra} 
for indices $i, j$ are preserved
\begin{equation*}
\begin{aligned}
\vmu : \Cat(\ell=i) = \sigma(\traj^1_\ell, \traj^1),&&
\vnu : \Cat(\ell=j) = \sigma(\traj^2_\ell, \traj^2)
\end{aligned}
\end{equation*}
Here we overload the notations $\vmu$ and $\vnu$ also to refer to their corresponding \textit{probability vectors}.

To this end, our method is fully specified by providing an algorithm to sample $ (i, j) $ such that $ i \sim \vmu $ and $ j \sim \vnu $.
The collection of such joint distributions for $(i,j)$ are couplings of $\vmu$ and $\vnu$, i.e. $\Gamma(\vmu, \vnu)$.

\subsection{Optimal transport couplings}

To repeat, our aim is to construct coupled kernels in which the coupled chains from \algref{alg:coupled_chains} meet in a relatively small number of MCMC steps, i.e. short meeting time.
Unfortunately, it is not clear how to directly minimize the meeting time. 
Intuitively, one might expect a kernel which, informally, ``brings chains closer'' to also have an improved meeting time. Naturally this brings us to consider the following problem:
\begin{equation}
\gamma := \argmin_{\gamma'} \sum_{i,j} \gamma_{ij}' D_{ij} \;\text{ s.t.~}\; \gamma' \in \Gamma(\vmu, \vnu)
\label{eq:ot_objective}
\end{equation}
where $D_{ij} = d(\traj_i^1, \traj_j^2)$ is the distance in the \textit{position space} between the $i$-th point in the first trajectory and the $j$-th point in the second. 
This is an example of a \emph{Kantorovich problem}, a well-studied family of problems from optimal transport \citep{villani2003topics}. 
In the case where $d_2^p(x, y) = \norm{x - y}_2^p$, we will refer to the minimizer as the \emph{$W_p$-coupling} due to the role it plays in the Wasserstein distance wrt. Euclidean metric $W_p(\vmu, \vnu) = \left( \inf_{\gamma \in \Gamma(\vmu, \vnu)} \mathbb{E}_{(X, Y) \sim \gamma} \norm{x - y}_2^p \right)^{1/p}$.

In this work we will consider two different choices for the metric $d$: 1) Euclidean distance $d_2$ which gives rise to the \textit{$W_2$-coupling}, and 2) 0-1 distance $d_I$ which gives rise to the \textit{maximal coupling}.

\subsection{$W_2$-coupling}
Arguably the most natural choice of metric $d$ in~\eqref{eq:ot_objective} is the squared Euclidean distance $d_2^2(\traj_i^1, \traj_j^2) = \norm{q_i^1 - q_j^2}_2^2$, whose solution we denote $\gamma^\circ$. Once we obtain $\gamma^\circ$, sampling $(i, j)$ is straightforward.
For completeness, we provide the full algorithm as Algorithm~\ref{alg:joint_sampling} in Appendix~\ref{app:method}.

Computationally, 
the optimization in~\eqref{eq:ot_objective} can be solved by generic linear programming solvers or more specialized methods, 
e.g. as in \citet{bonneel2011displacement}.
Such solvers in general have a time complexity $\calO(K^3)$ where $K$ is the length of the probability vectors $\vmu$ and $\vnu$.
This can be alleviated by using an approximate solver which could introduce biases.
Therefore, similarly to \citet{jacob16_coupl_partic_filter}, we also describe a debiasing method that allows the use of approximate solvers in \appref{app:method}.

\subsection{Maximal coupling}
For general choices of $d$~\eqref{eq:ot_objective} we do not have analytical solutions,
but for the particular choice $d_{I}(\traj_i^1, \traj_j^2) = \ind(i \neq j)$ we do.
In this case, the solution is the well-known maximal coupling $\gamma^\ast$ of two categorical distributions, which can be represented in its mixture view as
\begin{eqnarray}
\gamma^\ast =
\omega \frac{\vmu \wedge \vnu}{Z} + (1 - \omega) \frac{\vmu - (\vmu \wedge \vnu)  + \vnu - (\vmu \wedge \vnu)}{1 - Z}
\label{eq:maximal_coupling}
\end{eqnarray}
where $\wedge$ is the point-wise minimum operation,
$\omega = \PP(i=j)$ and
$Z = \sum_i (\vmu \wedge \vnu)_i $;
sampling from $\gamma^\ast$ is therefore tractable and straightforward. 
The process is summarized in Algorithm~\ref{alg:maximal_coupling} in \appref{app:method}.

By definition, for a maximal coupling $\gamma^\ast$
the probability of choosing pairs with the same time-index in two trajectories is maximized; we refer to such pairs with same indices as "index-aligned" pairs.
As we will see in \secref{sec:theoretical}, this property allows us to exploit Lemma 1 in \citet{heng_unbiased_2019} to show that the distance between the two coupled chains decreases with non-zero probability when the potential is strongly convex, or, equivalently, the target is strongly log-concave.

Though the idea of index-aligned pairs is useful to establish the theoretical results, it is not necessarily so in practice.
Note that as the approximation of Hamiltonian simulation by numerical integration becomes more accurate when step sizes become smaller, the joint $\gamma^\ast$ converges to the diagonal uniform distribution, i.e. $\gamma_{ii} \approx 1 / K$ and $\gamma_{ij} \approx 0$ for $i \ne j$ for large $K$. It is easy to construct examples where this leads to sub-optimal behavior when the goal is to minimize distance between the proposed states; \figref{fig:illustration} illustrates this nicely.

\subsection{An illustration of different couplings}
\label{sec:illustration}

We now illustrate how different intra-trajectory couplings behave using a 2D Gaussian with zero mean and unit diagonal covariance.
We start by simulating two Hamiltonian trajectories from $q^1_0=[0.5, 2.0]$ and $q^2_0=[0.5, -1.0]$ using the same momentum $p_0=[1.0, 1.0]$ for 7 steps, 
obtaining two trajectories $\traj^1$ and $\traj^2$ in Figure~\ref{fig:illustration-trajectory},
where the arrows represent the initial momentum $p_0$.
\begin{figure}[t]
    \ffigbox[\textwidth]{%
    \begin{subfloatrow}
        \ffigbox[0.8\textwidth]
        {\caption{Coupled trajectories}\label{fig:illustration-trajectory}}
        {\includegraphics[trim={0 1.0cm 0 1.0cm},clip,width=\linewidth]{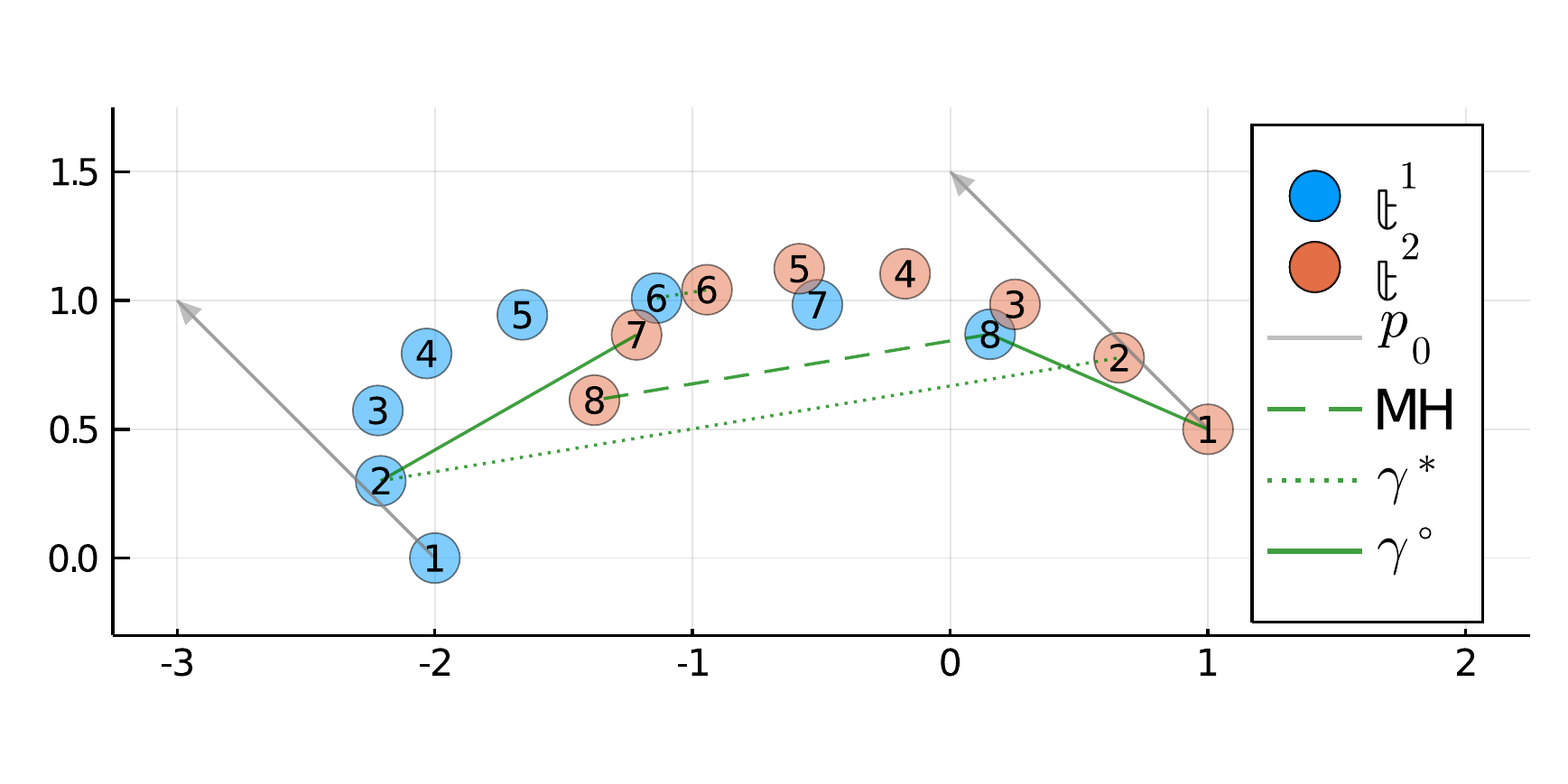}}
    \end{subfloatrow}
    \begin{subfloatrow}[2]
    \ffigbox[0.4\textwidth]
    {\caption{Maximal coupling}\label{fig:illustration-maximal}}
    {\includegraphics[trim={0 1.0cm 0 0.5cm},clip,width=\linewidth]{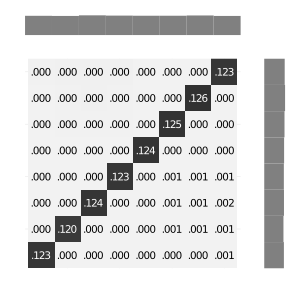}}
    \ffigbox[0.4\textwidth]
    {\caption{$W_2$-coupling}\label{fig:illustration-ot}}
    {\includegraphics[trim={0 1.0cm 0 0.5cm},clip,width=\linewidth]{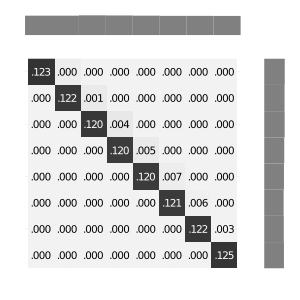}}
    \end{subfloatrow}
    }{%
    \caption{An illustration of different HMC couplings. Green lines in Figure~\ref{fig:illustration-trajectory} indicate possible pairs from different methods. For coupled Metropolis HMC, the dashed line pairs the end-points of two trajectories, which has a relative large distance. The dotted line is for multinomial HMC with maximal coupling. Though there is a change that the $6$-th points of two trajectories are paired, other index-aligned pairs are equally likely (Figure~\ref{fig:illustration-maximal}). E.g. the pair of $2$-th points has a large distance. In contrast, all pairs from multinomial HMC with $W_2$-coupling (solid lines) have relatively small distances, resulting in a small distance on average. 
    To see this, we calculate the expected distances: they are $1.37$ for $W_2$-coupling and $1.97$ for maximal coupling, where the former is clearly smaller, as expected.
    }
    \label{fig:illustration}
    }
    \vspace{-1.0em}
\end{figure}
We then sample from our two couplings to generate $100,000$ pairs of indices to estimate the joint distributions and to compute the marginals, which are shown in Figure~\ref{fig:illustration-maximal} and Figure~\ref{fig:illustration-ot}.
Note how the joint distributions differs: the ordering of the pairings are ``reversed''.
This intuitively makes sense when looking at Figure~\ref{fig:illustration-trajectory},
in which, e.g. the closest point for $\traj^1_1$ is $\traj^2_8$.
Finally, note that this U-turn example is chosen to highlight the differences between the two couplings.
If there was no U-turn, the differences between the two couplings could potentially be smaller.
\section{Theoretical analysis}
\label{sec:theoretical}
We now establish geometric tails for the meeting time for the mixture kernel in \eqref{eq:mixture_kernel} with the proposed coupled HMC kernels as the HMC component, thus satisfying the necessary conditions to use the estimator \eqref{eq:H_k}.

\paragraph{Proof sketch}
To prove geometric tails it turns out that it is sufficient to prove that the methods satisfy the conditions for Proposition 1 in \citet{heng_unbiased_2019}. 
Informally, the proposition states that once the chains enter a region $S$ in the state space where the target density is strongly log-concave, there is a non-zero probability that the chains will end up in a $\delta \text{-neighbourhood}$ of each other in some $n_0 \in \mathbb{N}$ steps.
The proof presented in \citet{heng_unbiased_2019} obtains this statement for the coupled Metropolis HMC by arguing directly about the probabilities of such an event conditioned on the initial states being in $S$. 
Here we instead prove a slightly stronger property, \emph{local contractivity}, from which the proposition follows immediately.
Informally, local contractivity ensures that the distance between the chains will decrease on average when initialized in some region.
We first prove that this holds for the maximal coupling by exploiting the fact that it maximizes the probability of picking index-aligned pairs, which, as mentioned before, is guaranteed to decrease the distance compared to the initial positions for strongly log-concave targets.
Once this has been established, local contractivity for the $W_2$-coupling follows immediately since by definition $W_2$-coupling has a smaller expected distance than the maximal coupling.
The remainder of the proof is essentially identical to \citet{heng_unbiased_2019} where excursions from the set $S$ is controlled with a geometric drift condition, from which we obtain geometric tails for the meeting time and thus validity of the methods.

Following \citet{heng_unbiased_2019}, we make two assumptions on the potential function $U: \RR^d \mapsto \RR $. %
\begin{assumption}[Regularity and growth of potential]
  The potential $U$ is twice continuously differentiable and its gradient $\nabla U: \RR^d \mapsto \RR^d$ is globally $\beta$-Lipschitz, i.e. there exists $\beta > 0$ such that
  $\norm{\nabla U(q) - \nabla U(q')} \leq \beta \norm{q - q'}$
  for all $q, q' \in \RR^d$.
  \label{ass:regularity_and_growth}
\end{assumption}
\begin{assumption}[Local strong convexity of potential]
  There exists a compact set $S \in \calB(\RR^d)$, with positive Lebesgue measure, s.t.~the restriction of the potential $U$ to $S$ is $\alpha$-strongly convex, i.e.,
  $\exists\;\alpha > 0$ s.t.~$(q - q')^\top \left( \nabla U(q) - \nabla U(q') \right) \geq \alpha \norm{q - q'}^2$
  for all $q, q' \in S$.
  \label{ass:local_strong_convexity}
\end{assumption}

Unless otherwise specified, we will let $S$ denote the set in \assumptref{ass:local_strong_convexity}, $\KerCHMC_{\stepsize{}, L}^{\gamma}$ denote a coupled HMC kernel as described in \algref{alg:coupled-hmc} with (i) shared momentum, (ii) shared forward and backward simulation steps and (iii) $(i, j) \sim \gamma$ for intra-trajectory sampling, and $\PkerCHMC{\stepsize{}}{L}{\gamma}$ denote the law of a coupled HMC kernel $\KerCHMC_{\stepsize{}, L}^{\gamma}$. For functions $f: \mathbb{R}^d \to \mathbb{R}$, we will also use the notation $L_{\ell}(f) = \left\{ x \in \mathbb{R}^d : f(x) \le \ell \right\}$ for the levelsets of $f$ and $f_A := f \big|_{A}$ for the restriction of $f$ to $A \in \mathcal{B}(\mathbb{R}^d)$.

\subsection{Geometric tails via local contractivity}
\label{sec:theory-tails}

We first state the definition of local contractivity and Proposition 1 from \citet{heng_unbiased_2019}.

\begin{condition}[Local contractivity]
    Given a compact set $ S \in \calB(\RR) $ with positive Lebesgue measure, we say the kernel $\CMKer_{\stepsize{}, L}^{\gamma}$ is \emph{locally contractive} on $S$ with rate $\rho \in (0, 1)$ if there exists $m \ge 1$ such that for any given $k_0 > 0$ there exists $\bar{\stepsize{}} > 0$, $\bar{L} \in \mathbb{N}$ s.t.
    \begin{equation}\label{eq:local_contractivity}
        \E[(l_1, l_2) \sim \gamma]{\norm{\flowqdisc{\stepsize{}}{l_1}(q^1, p) - \flowqdisc{\stepsize{}}{l_2}(q^2, p)}^m} \le \rho^m \norm{q^1 - q^2}^m
    \end{equation}
    for all $\stepsize{} \in (0, \bar{\stepsize{}})$, $L  \in \mathbb{N}$ such that $\stepsize{} L < \bar{\stepsize{}} \bar{L}$ and for all $(q^1, q^2, p) \in S \times S \times \levelset{k_0}{K}$.
    \label{ass:local_contractivity}
\end{condition}

Informally, this is saying that there exists a step size and integration time such that a single application of $\CMKer_{\stepsize{}, L}^{\gamma}$ decreases the distance between the two states on average. Furthermore, this property is preserved when decreasing the step size or the integration time. This last part is important since different parts of the analysis will require possibly smaller step sizes and/or integration times. Thus, by ensuring that all statements hold for all smaller step sizes and integration times, we can combine the statements by simply choosing the minimum of the step sizes and/or integration times required by the different statements.

\begin{proposition}[Proposition 1, \citet{heng_unbiased_2019}]\label{prop:proposition-1-heng2019}
Suppose that the potential $U$ satisfies Assumptions \ref{ass:regularity_and_growth} and \ref{ass:local_strong_convexity}.
Then for any $\delta > 0$, $u_0 > \inf_{q \in S} U(q)$, and $u_1 < \sup_{q \in S} U(q)$ with $u_0 < u_1$, there exists $\bar{\stepsize{}} > 0$ and $\bar{L} \in \mathbb{N}$ such that for any $\stepsize{} \in (0, \bar{\stepsize{}})$ and $L \in \mathbb{N}$ satisfying $\stepsize{} L < \bar{\stepsize{}} \bar{L}$, there exists $v_0 \in (u_0, u_1)$, $n_0 \in \mathbb{N}$ and $\omega \in (0, 1)$ such that
\begin{equation}\label{eq:proposition-1-heng2019}
    \inf_{q^1, q^2 \in S_0} \KerCHMC_{\stepsize{}, L}^{\gamma, n_0} \left( (q^1, q^2), D_{\delta} \right) \ge \omega
\end{equation}
where $S_0 = L_{v_0}(U_S)$ is compact with positive Lebesgue measure,
\begin{equation*}
\begin{aligned}
    & \KerCHMC_{\stepsize{}, L}^{\gamma, n} \left( (q^1, q^2), A^1 \times A^2 \right) = \\
    &\quad \PkerCHMC{\stepsize{}}{L}{\gamma} \left( (Q_n^1, Q_n^2) \in A^1 \times A^2 \mid (Q_0^1, Q_0^2) = (q^1, q^2) \right)
\end{aligned}
\end{equation*}
denotes the n-step transition probabilities of the coupled chain, and $D_{\delta} = \set{(q, q') \in \mathbb{R}^d \times \mathbb{R}^d : \norm{q - q'} \le \delta}$.
\end{proposition}
As noted earlier, this proposition is a key step in the proof of the coupled Metropolis HMC kernel in \citet{heng_unbiased_2019}. This ensures that once we reach the set $S_0 \subset S$, there is a non-zero probability that within some finite number of steps the chains will be $\delta$-close, i.e.~meet in the relaxed sense. If we can then also ensure that this set $S_0$ will be entered by the coupled chains sufficiently often, then we bound the tails of distribution over meeting times.

We now establish that indeed, for coupled multinomial HMC kernels, \condref{ass:local_contractivity} implies \propref{prop:proposition-1-heng2019}.

\begin{lemma}\label{lemma:local-contractivity-implies-prop-1}
If $\KerCHMC_{\stepsize{}, L}^{\gamma}$ satisfies \condref{ass:local_contractivity}, then $\KerCHMC_{\stepsize{}, L}^{\gamma}$ satisfies the conditions of \propref{prop:proposition-1-heng2019}.
\end{lemma}

\begin{prf}
    Observe that
    \begin{equation*}
    \begin{split}
      & \PkerHMC{\stepsize{}}{L} \left( \norm{Q_1^1 - Q_1^2} \le \rho \norm{Q_0^1 - Q_0^2} \mid (Q_0^1, Q_0^2) = (q^1, q^2) \right) \\
      &= \E[\KerCHMC_{\stepsize{}, L}^{\gamma}]{\ind \set{\norm{Q_1^1 - Q_1^2} \le \rho \norm{q^1 - q^2}}} \\
      &= \E[P \sim \mathcal{N}(0, I)]{\E[(l_1, l_2) \sim \gamma]{\ind(R_{q^1, q^2, P}) \mid P}}
    \end{split}
    \end{equation*}
    where we have let $R_{q^1, q^2, p}$ denote the set of events where we have contraction,~i.e.
    \begin{equation*}
      R_{q^1, q^2, p} = \set{\norm{\flowqdisc{\stepsize{}}{l_1}(q^1, p) - \flowqdisc{\stepsize{}}{l_2}(q^2, p)} \le \rho \norm{q^1 - q^2}}
    \end{equation*}
    By \condref{ass:local_contractivity}, $\E[(l_1, l_2) \sim \gamma]{\ind (R_{q^1, q^2, p})} > 0$
    for any $(q^1, q^2, p) \in S \times S \times \levelset{k_0}{K}$, where $k_0 > 0$ is to be decided, since otherwise~\eqref{eq:local_contractivity} would not hold. Hence a single application of the kernel $\KerCHMC_{\stepsize{}, L}^{\gamma}$ will have a non-zero probability of decreasing the distance between the states if we are in $S$. The remainder of the proof ensures that parameters can be chosen such that there is a non-zero probability of staying within the set $S_0 \subset S$ for some $n_0 := \inf \set{n \in \mathbb{N}: \rho^n \norm{q^1 - q^2} \le \delta}$ applications of the kernel, i.e. $(Q_k^1, Q_k^2) \in S_0 \times S_0$ for all $k = 1, \dots, n_0$. This finally allows us to conclude that there is a non-zero probability of entering $D_{\delta}$ if we are currently in the set $S_0$. 
    See the Appendix~\ref{app:proof-local-contractivity-implies-prop-1} for the full proof.
\end{prf}

\begin{theorem}[Theorem 2, \citet{heng_unbiased_2019}]\label{thm:theorem-2-heng2019}
    Suppose that the potential $U$ satisfies Assumptions \ref{ass:regularity_and_growth} and \ref{ass:local_strong_convexity}.
    Suppose that there exists $\bar{\stepsize{}} > 0$ and $\bar{\sigma}> 0$ such that for any $\stepsize{} \in (0, \bar{\stepsize{}})$, $L \in \NN$ and $\sigma \in (0, \bar{\sigma})$, there exists a measurable function $V: \mathbb{R}^d \to [1, \infty)$, $\lambda \in (0, 1)$, $b < \infty$ and $\mu > 0$ such that
    \begin{equation*}
    \KerHMC_{\stepsize{}, L}(V)(x) \le \lambda V(x) + b, \quad Q_{\sigma}(V)(x) \le \mu \left( V(x) + 1 \right)
    \end{equation*}
    for all $x \in \mathbb{R}^d$, $\pi_0(V) < \infty$, $\lambda_0 = (1 - \gamma) \lambda + \gamma (1 + \mu) < 1$ and $\set{x \in \mathbb{R}^d: V(x) \le \ell_1} \subseteq \set{x \in S: U(x) \le \ell_0}$, for some $\ell_0 \in \set{\inf_{x \in S} U(x), \sup_{x \in S} U(x)}$ and $\ell_1 > 1$ satisfying $\lambda_0 + 2 \left( (1 - \gamma) b + \gamma \mu \right) (1 - \lambda_0)^{-1} (1 + \ell_1)^{-1} < 1$.
    Then there exists $\stepsize{}_0 \in (0, \bar{\stepsize{}})$, $L_0 \in \mathbb{N}$ and $\sigma_0 > 0$ such that for any $\stepsize{} \in (0, \stepsize{}_0)$, $L \in \mathbb{N}$ satisfying $\stepsize{} L < \stepsize{}_0 L_0$ and $\sigma \in (0, \sigma_0)$, we have
    \begin{equation*}
        \PkerCMHMC{\stepsize{}}{L}{\sigma}{\gamma}(\tau > n) \le C_0 \kappa_0^n
    \end{equation*}
    for some $C_0 \in \mathbb{R}_{ + }$ and $\kappa_0 \in (0, 1)$ and for $n \in \mathbb{N}_0$, where $\PkerCMHMC{\stepsize{}}{L}{\sigma}{\gamma}$ denotes the law of the kernel $\KerCHMC_{\stepsize{}, L, \sigma}^{\gamma}$ for a given coupled HMC kernel $\KerCHMC_{\stepsize{}, L}^{\gamma}$.
    \label{theorem:geometric_tails}
\end{theorem}
\begin{prf}
    The proof is identical to \citet{heng_unbiased_2019} via \lemref{lemma:local-contractivity-implies-prop-1}.
\end{prf}

\subsection{Local contractivity for $W_2$-coupling and maximal coupling}
\label{sec:theory-local_contractivity}
Now we
establish \condref{ass:local_contractivity} for coupled multinomial HMC kernels with maximal coupling $\gamma^\ast$ and $W_2$-coupling $\gamma^\circ$, ensuring that \thmref{theorem:geometric_tails} applies to the resulting mixture kernels $\KerCHMC_{\stepsize{}, L, \sigma}^\ast$ and $\KerCHMC_{\stepsize{}, L, \sigma}^\circ$.

We first restate a slight variation of Lemma 1 from \citet{heng_unbiased_2019}, which tells us that the states reached by \emph{exact} flows with shared momentum is closer than the initial states for sufficiently small integration times.\looseness=-1

\begin{lemma}\label{lemma:local_abs_contractivity}
    Suppose that the potential $U$ satisfies Assumptions \ref{ass:regularity_and_growth} and \ref{ass:local_strong_convexity}. For any compact set $A \subset S \times S \times \mathbb{R}^d$, there exists a trajectory length $T > 0$ such that
    \begin{equation}\label{eq:local_abs_contractivity}
        \norm{\flowq{t}(q^1, p) - \flowq{t}(q^2, p)} \leq \rho \norm{q^1 - q^2}
    \end{equation}
    for all $t \in [-T, T] \setminus \set{0}$ and all $(q^1, q^2, p) \in A$.
\end{lemma}

\begin{prf}
See \appref{app:proof-local_abs_contractivity} for the detailed proof.
\end{prf}

Note that \lemref{lemma:local_abs_contractivity} is a statement about the distance between the integrated states \emph{at the same integration time $t$}. 
As an immediate consequence the expected distance with respect to a joint distribution with probability mass only along the diagonals satisfies a similar property, controlling for numerical errors (see \appref{app:proof-property-parallel}).
Therefore our strategy in proving local contractivity for $\CMKer_{\stepsize, l}^\ast$ and $\CMKer_{\stepsize, l}^\circ$ is to ensure that as we decrease the stepsize the probability mass on the diagonals, i.e. $\mathbb{P}(i = j)$, can be made close to 1.

\paragraph{Maximal coupling}
\label{sec:theory-local_contractivity-maximal}
To establish \condref{ass:local_contractivity} for coupled multinomial HMC with maximal coupling $\gamma^\ast$, 
$\CMKer_{\stepsize, l}^\ast$,
we first introduce a bound on the total variation distance between the trajectory distributions $\vmu$ and $\vnu$.
\begin{proposition}\label{prop:maximal_coupling}
    Suppose that $U$ satisfies Assumptions \ref{ass:regularity_and_growth} and \ref{ass:local_strong_convexity}.
    For any $\delta > 0$, there exists $\stepsize{}_0 > 0$, $L_0 \in \mathbb{N}$ s.t.~for all $\stepsize{} \in (0, \stepsize{}_0)$, $L \in \mathbb{N}$ satisfying $\stepsize{} L < \stepsize{}_0 L_0$, we have
    \begin{equation}
        \TV( \vmu, \vnu ) = \PP(i \neq j) < \delta.
    \end{equation}
\end{proposition}
\begin{prf}
    The proof uses the 1-Lipschitz property of the softmax function \citep{gao2018properties} in \eqref{eq:multi-intra} to bound the probability differences introduced by numerical errors.
    See Appendix~\ref{app:proof-maximal_coupling}.
\end{prf}
With Proposition~\ref{prop:maximal_coupling}, 
we can establish local contractivity for coupled multinomial HMC kernels with $\gamma^\ast$.
\begin{lemma}\label{lemma:contractivity-maximal}
    $\CMKer_{\stepsize, l}^\ast$ satisfies \condref{ass:local_contractivity}.
\end{lemma}
\begin{prf}
    Due to Proposition~\ref{prop:maximal_coupling}, for a given integration time, we can choose step size arbitrary small to increase probability of picking parallel-in-time pairs, whose contractivity is established in Proposition~\ref{prop:property-parallel}.
    See Appendix~\ref{app:proof-contractivity-maximal} for the complete proof.
\end{prf}

\paragraph{$W_2$-coupling}
\label{sec:theory-local_contractivity-ot}
Similarly to Lemma~\ref{lemma:contractivity-maximal}, for coupled multinomial HMC with $\gamma^{\circ}$, $\CMKer_{\stepsize, l}^\circ$, we have:
\begin{lemma}\label{lemma:contractivity-ot}
    $\CMKer_{\stepsize, l}^\circ$ satisfies \condref{ass:local_contractivity}.
\end{lemma}
\begin{prf}
    The definition of $\gamma^\circ$ from \eqref{eq:ot_objective} implies
    \begin{equation}
    \begin{aligned}
        & \E[(i, j) \sim \gamma^\circ] { \norm{\flowqdisc{\stepsize{}}{l_i}(q^1, p) - \flowqdisc{\stepsize{}}{l_j}(q^2, p)}^2 } \leq \\
        & \quad \E[(i, j) \sim \gamma^\ast] { \norm{\flowqdisc{\stepsize{}}{l_i}(q^1, p) - \flowqdisc{\stepsize{}}{l_j}(q^2, p)}^2 }
    \end{aligned}
    \end{equation}
    The rest of the proof follows that of \lemref{lemma:contractivity-maximal}.
\end{prf}
Lemmas \ref{lemma:contractivity-maximal} and \ref{lemma:contractivity-ot} together with \thmref{thm:theorem-2-heng2019} then establishes geometric tails for the meeting time of the resulting mixture kernels $\CMKer_{\stepsize, l}^\ast$ and $\CMKer_{\stepsize, l}^\circ$, respectively.

\section{Experiments}
\label{sec:exp}

In this section, we evaluate the performance of the proposed coupled HMC kernels.
The coupled Metropolis HMC by \citet{heng_unbiased_2019} is used as a baseline. 
Following \citet{heng_unbiased_2019},
we combine coupled HMC kernels with a coupled RWMH kernel to obtain exact couplings:
we take the standard deviation of the RWMH kernel to be $\sigma = 10^{-3}$ and the mixture cofficient (i.e. probability of using RWMH) to be $\alpha = 1/20$. 
For the estimation task,
we consider a more efficient (still unbiased) \textit{time-averaged} variant of \eqref{eq:H_k}:\looseness=-1
\begin{equation}
    H_{k:m}(X, Y) = \frac{1}{m - k + 1} {\sum}_{i=k}^m H_{i}(X, Y)
\label{eq:H_km}
\end{equation}
Furthermore, we run $R$ independent pairs of coupled chains $(X^r, Y^r)$, $r = 1, \dots, R$, and estimate $H^\dagger$ as $\hat{H} = R^{-1} \sum_{r=1}^R H_{k:m} (X^{(r)}, Y^{(r)})$. 
For $x \in \mathbb{R}^d$, we consider $h$ as the first and second moments of $x$, i.e.~ $h_i(x) = x_i$ and $h_{d+i}(x) = x_i^2$ for $i = 1, \dots, d$.

We consider three target distributions.
The first target is a \textbf{1{,}000D Gaussian}.
The second one is the posterior of a \textbf{Bayesian logistic regression} model on the German credit dataset \citep{asuncion2007uci}.
We apply the same pre-processing as in \citet{heng_unbiased_2019},
which results in a sampling space of $\mathbb{R}^{302}$.
The last model considered is a \textbf{log-Gaussian Cox point process} that models tree locations in a forest.
We discretize the forest using a $16 \times 16$ grid%
, resulting in a sampling space on $\mathbb{R}^{256}$.
Note that the first two targets meet necessary conditions from Section~\ref{sec:theoretical}.
More details of these targets can be found in \appref{app:target_distributions}.

Our implementation is based on \textsc{AdvancedHMC.jl} \citet{xu2020advancedhmc} and is available at at \url{https://github.com/TuringLang/CoupledHMC.jl}, which also contains scripts to reproduce results in this paper.

\subsection{Meeting time comparisons}
\label{sec:meeting}

We first investigate how the meeting time $\tau$ of our method changes under different step sizes $\epsilon$ and numbers of leapfrog steps $L$. For this purpose, we run all coupled HMC methods
initialised at a random draw from $\mathcal{N}(0, I)$ for 1,000 iterations.
For each method, we use different step sizes $\stepsize{}$ and leapfrog steps $L$: $(\stepsize{}, L) \in \{0.05, 0.07, \dots, 0.45\} \times \{5, 10, 15\}$ for $1{,}000$D Gaussians, $(\stepsize{}, L) \in \{0.01, 0.0125, \dots, 0.04\} \times \{10, 20, 30\}$ for logistic regression and $(\stepsize{}, L) \in \{0.05, 0.07, \dots, 0.45\} \times \{10, 20, 30\}$ for log-Gaussian Cox point processes. Furthermore, we repeat each experiment for $R=10$ times to estimate standard derivation.
Figure~\ref{fig:meeting} shows resulting meeting time together with standard derivation for varying $\stepsize{}$ and a fixed $L=10$; figures for other $L$ are similar so we defer those to Appendix~\ref{app:sweep}.
\begin{figure*}[t]
    \ffigbox[\textwidth]{%
        \centering
        \begin{tikzpicture}
    \begin{axis}[%
        hide axis, xmin=10, xmax=50, ymin=0, ymax=0.4,
        legend style={color={rgb,1:red,0.1333;green,0.1333;blue,0.3333}, draw opacity={0.1}, line width={1}, solid, fill={rgb,1:red,1.0;green,1.0;blue,1.0}, fill opacity={0.9}, text opacity={1.0}, font={{\fontsize{8 pt}{10.4 pt}\selectfont}}, at={(1.02, 1)}, anchor={north west}, legend columns=-1}
    ]
        \addlegendimage{color={rgb,1:red,0.2667;green,0.4667;blue,0.6667}, name path={41d92112-fced-469f-bc61-94814c6a6591}, draw opacity={0.7}, line width={1.2}, solid, mark={diamond*}, mark size={2.25 pt}, mark options={color={rgb,1:red,0.0;green,0.0;blue,0.0}, draw opacity={0.7}, fill={rgb,1:red,0.2667;green,0.4667;blue,0.6667}, fill opacity={0.7}, line width={0.0}, rotate={0}, solid}}
        \addlegendentry{Metropolis}
        \addlegendimage{color={rgb,1:red,0.1333;green,0.5333;blue,0.2}, name path={d1056c54-bf38-4a84-9416-9b3b5e66f354}, draw opacity={0.7}, line width={1.2}, solid, mark={star}, mark size={2.25 pt}, mark options={color={rgb,1:red,0.0;green,0.0;blue,0.0}, draw opacity={0.7}, fill={rgb,1:red,0.1333;green,0.5333;blue,0.2}, fill opacity={0.7}, line width={0.0}, rotate={0}, solid}}
        \addlegendentry{Maximal}
        \addlegendimage{color={rgb,1:red,0.8;green,0.7333;blue,0.2667}, name path={98ad80d3-1b99-4400-911a-13702f6a1df5}, draw opacity={0.7}, line width={1.2}, solid, mark={*}, mark size={2.25 pt}, mark options={color={rgb,1:red,0.0;green,0.0;blue,0.0}, draw opacity={0.7}, fill={rgb,1:red,0.8;green,0.7333;blue,0.2667}, fill opacity={0.7}, line width={0.0}, rotate={0}, solid}}
        \addlegendentry{$W_2$}
    \end{axis}
\end{tikzpicture}
        \begin{subfloatrow}[3]
            \ffigbox[0.31\textwidth]
            {\caption{1000D Gaussian}}
            {\includegraphics[trim=0 0.25cm 0 0.1cm,clip,width=\linewidth]{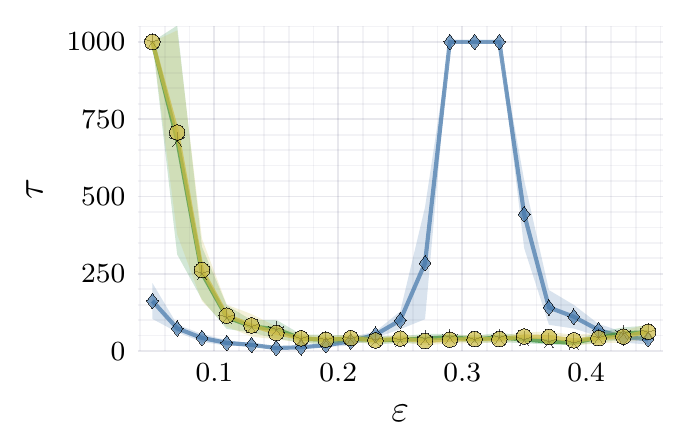}\label{fig:meeting-gaussian}}
            \ffigbox[0.31\textwidth]
            {\caption{Logistic regression}}
            {\includegraphics[trim=0 0.25cm 0 0.1cm,clip,width=\linewidth]{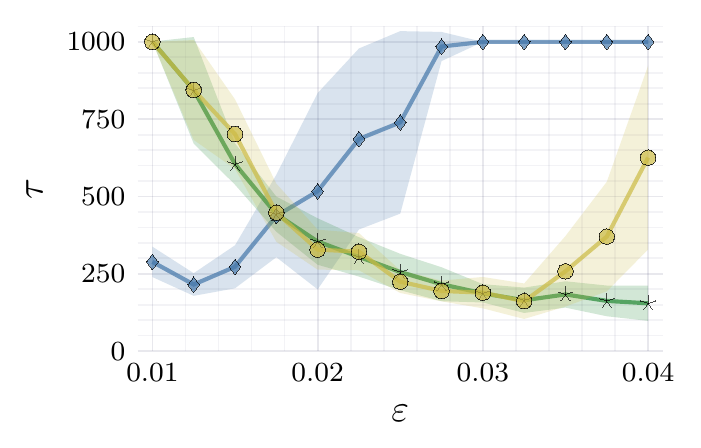}\label{fig:meeting-lr}}
            \ffigbox[0.31\textwidth]
            {\caption{Log-Gaussian Cox point process}\label{fig:meeting-coxprocess}}
            {\includegraphics[trim=0 0.25cm 0 0.1cm,clip,width=\linewidth]{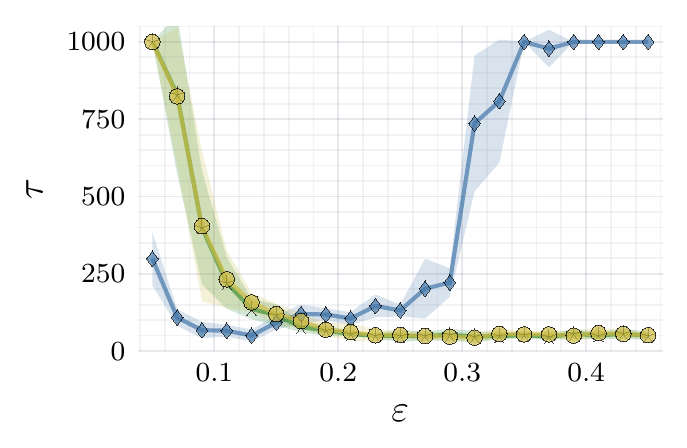}}
        \end{subfloatrow}
    }{%
        \caption{Meeting time $\tau$ with different $\stepsize{}$ and $L=10$ out of $R=10$ runs with lines for average and shade for 1 standard deviation. Overall, coupled multinomial HMC attains smaller meeting time and is more robust to $\stepsize{}$.}
        \label{fig:meeting}
    }
\vspace{-1em}
\end{figure*}
It is worth noting that $\tau$ equal to 1,000 should be interpreted as coupled chains \textit{did not meet within 1,000 iterations}.

Figure~\ref{fig:meeting} shows clearly that both maximal coupling and $W_2$-coupling achieve smaller meeting time than the baseline for large step sizes.
This robustness against large step sizes is useful in practice 
since it allows us to simulate a trajectory of a given length with less computation, by using larger $\epsilon$ rather than larger $L$.
However, when the step size is sufficiently small Metropolis HMC will almost always accept the end-point, thus travel the full integration length $T$ at every step.
In contrast, multinomial HMC will put uniform mass on intermediate states which means that it travels $1/4 T$ in expectation. 
Therefore Metropolis HMC will move towards the typical set faster and thus have a smaller meeting time compared to multinomial HMC. 
It is also worth noting the surge in meeting time for coupled Metropolis HMC in \figref{fig:meeting-gaussian} around $\stepsize{}=0.3$ can be explained by the similar phenomena observed in \figref{fig:illustration}, large trajectory length can lead to end-points close to their starting points, thus never meet.
In particular, the trajectory length $3 = 0.3\times10$ is around $\pi \approx 3.14$, in which case trajectories are basically full circles ending close to where they start, a special case for Gaussians.

Furthermore, for logistic regression (Figure~\ref{fig:meeting-lr}),  optimal parameters of HMC ($\epsilon=0.03, L=10$) leads to excessively long meeting time. This result is consistent with those in \citet{heng_unbiased_2019}. 
This is clearly undesirable since optimal parameters for good sampling efficiency leads to non-contractive coupled chains.
Besides, it is worth noting that maximal coupling is more robust to large step sizes than $W_2$-optimal coupling for the logistic regression model.
To understand this, recall that $W_2$-coupling takes a local greedy approach but there is \textit{no guarantee} it can lead to faster meeting through multiple transitions. With large step sizes, numerical errors in simulation are enlarged, leading to more probabilities assigned to non-diagonal entries in the coupling matrix, equivalently more freedom in the $W_2$-coupling. In such cases, the greedy effect of $W_2$-coupling is also enlarged but such greedy approach \textit{turns out} to be less effective than maximal coupling for the logistic regression model, a target that satisfies Assumptions \ref{ass:regularity_and_growth} and \ref{ass:local_strong_convexity}. In short, whether or not the greedy approach is preferable is target-dependent. Specifically, when a target satisfies Assumptions \ref{ass:regularity_and_growth} and \ref{ass:local_strong_convexity}, one would expect maximal coupling to work well enough; when such assumptions fail, $W_2$-coupling can be more efficient, as seen in Appendix \ref{app:toy}.

Finally, as motivated earlier, one can use existing adaption techniques to choose parameters $\stepsize{},L$. For example, one can use the adapted parameters from preliminary runs of NUTS. As an concrete example, NUTS-adapted $\stepsize{},L$ for logistic regression are 0.022 and 22, and those for log-Gaussian Cox point processes are 0.28 and 16, which allows our method to meet relatively fast: 114 and 118 for the first model and 50 and 51 for the second one, for the two proposed kernels respectively.%

\subsection{Estimator efficiency comparisons}

Although Monte Carlo estimates by \eqref{eq:H_km} are unbiased,
it can have large variances due to the use of coupled but often short Markov chains.
In other words, making \eqref{eq:H_km} bias-free comes at a cost of increased variance.
Therefore, it is helpful to study the efficiency, or \textit{inefficiency}, of the estimator under a joint effect of removed bias but increased variance, which we define next.

For a vector-valued function $h$,
the variance of estimator $\hat{H}$ for coupled HMC is defined as $\sum_d \nu(h_d)$ where $\nu(h) = \mathbb{V}_r\left(H_{k:m}(h, X^r, Y^r)\right)$. Here $r$ is the index of repeated runs.
\emph{Asymptotic inefficiency} is defined as $\sum_d i(h_d)$ where $i(h) = \hat{C} \nu(h)$ \citep{glynn1992asymptotic}. 
Here $\hat{C} = \E[r]{2(\tau^r - 1) + \max(1, m + 1 - \tau^r)}$ is the \emph{expected cost} over $R$ runs and $\tau^r$ is the meeting time for the $r$-th run.
The asymptotic variance of (non-coupled) HMC can be approximated with the \texttt{spectrum0.ar} function of the \texttt{coda} R package \citep{plummer_coda_2006} using a long chain: 10,000 samples after a burn-in of 1,000 using $(\epsilon, L) = (0.03, 10)$ for logistic regression, and $(0.3, 10)$ for the log-Gaussian Cox point process model.
Relative inefficiency is then defined as the ratio of asymptotic inefficiency over asymptotic variance.

We study this inefficiency using logistic regression and log-Gaussian Cox point processes due to their wide adoption in practice.
Following \citet{heng_unbiased_2019}, we set $\stepsize{}$ and $L$
to those leading to smallest meeting time in Section~\ref{sec:meeting}.
We first perform 100 preliminary runs of coupled HMC kernels to get an empirical distribution of meeting time $\tau$. 
Then we use this distribution to determine $k$ and $m$ following heuristics from \citet{heng_unbiased_2019}: we
take $k$ as either the median or the 90\% sample quantile of $\tau$ and take $m$ as a product of $k$ and a constant, e.g. $5 k$ or $10 k$.
We then perform $R = 100$ independent runs of coupled chains with different combinations of $k$ and $m$ -- we expect a longer chain to have a smaller variance but with a larger computation budget.
Also bear in mind that an ideal asymptotic inefficiency should be close to 1. 

Table~\ref{tab:biasvariance} shows asymptotic inefficiencies for varying $k$ and $m$.
\begin{table}[t]
\centering
\begin{tabular}{cc|c|c|c}
\toprule
$k$ & $m$ & Metropolis\tablefootnote{We kindly note that relative inefficiencies reported here for Metropolis on logistic regression are different from \citet{heng_unbiased_2019} by a factor around $2.0$. This is due to a small mistake in their paper, as confirmed by the authors.} & Maximal & $W_2$ \\ \hline \hline
\multirow{2}{*}{median} & $5k$ & 2.40 & 17.69 & 3.18 \\
 & $10k$ & 2.39 & 7.37 & 3.73 \\
\multirow{2}{*}{90\% quantile} & $5k$ & 2.36 & 2.14 & 2.88 \\
 & $10k$ & 2.32 & 1.90 & \textbf{0.94} \\ \hline \hline
\multirow{2}{*}{median} & $5k$ & 6.03 & 4.84 & 7.62 \\
 & $10k$ & 4.81 & 4.01 & 6.00  \\
\multirow{2}{*}{90\% quantile} & $5k$ & 5.26 & 4.58 & 6.86 \\
 & $10k$ & 4.61 & \textbf{3.83} & 5.80  \\ \bottomrule
\end{tabular}
\caption{Relative inefficiency with different $k$ and $m$ for logistic regression (top half) and log-Gaussian Cox point processes (bottom half). Bold indicates the one most close to 1. Note that for each method, $k$ is different thus inefficiencies across different coupled kernels (across columns) are \textit{not directly comparable}. Instead, we aim to study if the relative inefficiency can be made close to $1$ with suitable parameters (across rows).}
\label{tab:biasvariance}
\vspace{-0.5em}
\end{table}
For logistic regression, with suitable choices of $k,m$, the relative inefficiency can be made close to 1 ($W_2$-coupling); for the other model, the best ($3.83$) is attained by maximal coupling, both of which are superior to best of Metropolis.
Overall, it demonstrates that the optimal transport couplings can achieve better bias-variance trade-off than the baseline \emph{when suitable $k$ and $m$ are chosen}.
Note that both optimal transport couplings \textit{seem} to be inefficient for $m=5k$ with $k$ being the median. This is because the coupled chains meet in much shorter time compared to Metropolis. As a result, the chains are not as close to the stationary distribution.
The table also confirms that a larger $m$ helps reduce the asymptotic inefficiency, at the cost of more computation. 
Fortunately, this variance reduction can \textit{also} be achieved by parallel coupled HMC chains.

\section{Related Work}

Research on couplings for MCMC methods has a long history~\citep{research2005coupled,johnson1996studying,johnson1998coupling,rosenthal1997faithful,meyn2012markov,rowland2018geometrically,nuesken2018constructing,jacob2019smoothing,biswas2019estimating}.
Couplings for HMC has been more recently focusing on Metropolis HMC, e.g.~ \citet{neal2017circularlycoupled}.
The work that is most closely related to ours is \citet{heng_unbiased_2019}, in which coupling for Metropolis HMC is established.
\citet{bou-rabee_coupling_2019} studied the convergence of Metropolis HMC, and also proposed a new way to couple momentum variables, called \textit{contractive coupling}, that does not rely on sharing them.\footnote{We also study the effect of contractive coupling in our methods empirically but defer this to Appendix~\ref{app:toy}.}

Our analysis in Section~\ref{sec:theoretical} is also related to works on convergence analysis of HMC on log-concave targets \citep{mangoubi2017rapid,chen2019optimal}.

\section{Conclusion}

In this paper,
we develop two novel couplings for multinomial HMC. %
We provide theoretical analysis on the validity of the proposed methods,
and perform simulations to 
demonstrate their advantages over existing methods in terms of meeting time and robustness to HMC parameters, which is an important step towards practical use of coupled HMC.
We hope this work will help advance the research on searching more efficient coupled HMC methods, and a wider use of coupled HMC for probabilistic modeling
in practice.
For future work, we are interested in extending the coupling methods to more advanced HMC variants, e.g. the no-U-turn algorithm \citep{hoffman_no-u-turn_2011}.

\section*{Acknowledgement}
Hong Ge and Tor Erlend Fjelde acknowledge generous support from Huawei Research and donations from Microsoft Research. The proposed algorithms' implementation is based on the Turing probabilistic programming language, from which the authors benefit greatly. We would also like to thanks Cameron Pfiffer, Mohamed Tarek, Martin Trapp, Sharan Yalburgi for helpful comments on an earlier draft of this paper. 
\bibliographystyle{apalike}
\bibliography{references,content/ref-hong}

\clearpage
\appendix

\section{Additional Background}
\label{app:background}

\subsection{Properties of Hamiltonian flow}

The flow map $ \flow{t} $ has the following properties:
\begin{enumerate}
\item (Reversibility). 
$\forall\;t \in \RR_+ $, 
the inverse flow map $ \Phi_t^{-1} $ satisfies $ \Phi_t^{-1} = R \circ \flow{t} \circ R $, 
where $ R(q, p) = (q, - p) $ denotes the momentum \underline{r}eversal operation.
\item (Energy conservation). 
The Hamiltonian $ \Energy $ of the system satisfies $ \Energy \circ\,\flow{t} = \Energy $.
\item (Measure preservation). 
For any $ t \in \RR_+ $ and $ A \in \calB(\RR^{2d})$, we have $ \Leb_{2d} \left( \flow{t}(A) \right) = \Leb_{2d} (A) $, where $ \Leb_d $ denotes the Lebesgue measure on $ \RR^d $.
\end{enumerate}
Together the properties ensures that the Markov kernel defined by the Hamiltonian flow leaves the extended target distribution $ \bar{\pi} $ invariant.

\subsection{Properties of leapfrog integration}

The numerical flow map $\flowdisc{\stepsize}{L}$ enjoys the following two inequalities due to the simplicity of order-two leapfrog integrators \citep{hairer2006geometric}
\begin{equation}\label{eq:numerical-error-state}
    \norm{\flowdisc{\stepsize}{L}(q_0, p_0) - \flow{\stepsize{} L}(q_0, p_0)} \leq C_a(q_0, p_0, L) \stepsize{}^2
\end{equation}
\begin{equation}\label{eq:numerical-error-energy}
    \norm{\Energy\left(\flowdisc{\stepsize}{L}(q_0, p_0)\right) - \Energy(q_0, p_0)} \leq C_b(q_0, p_0, L) \stepsize{}^2
\end{equation}
for some positive constants $C_a$ and $C_b$.
These two inequalities are used in several places through our theoretical analysis, e.g. in Section~\ref{app:proof-local_abs_contractivity} and Section~\ref{app:proof-maximal_coupling}.

\subsection{Coupled RWMH kernel}
The coupled RWMH kernel from \citet{heng_unbiased_2019} used in this paper is shown in Algorithm~\ref{alg:coupled-mh} for completeness. Note that here we slightly abuse notation, writing $\MKer_{\sigma}(X, Y)$ to mean denote the probability \emph{density} of the probability measure $\MKer_{\sigma}(X, \cdot)$ evaluated at $Y$, where $X$ and $Y$ are random variables.
\begin{algorithm}[t]
    \KwInput{A pair of current states $(X_0, Y_0)$ and a RWMH kernel $\MKer_\sigma$ with variance $\sigma^2 I_d$}
    \KwOutput{A pair of next states $(X', Y')$}
    Sample $X^\ast \sim \MKer_\sigma(X_0, \wc) $ \;
    Sample $ w \mid X \sim \Unif([0, \MKer_\sigma(X_0, X^\ast) ]) $ \;
    \eIf{$w\leq \MKer_\sigma(Y_0, X^\ast) $}
    {
        Set $Y^\ast=X^\ast$\;
    }
    {
        \Repeat{$w^\ast > \MKer_\sigma(X_0, Y^\ast)$}{
            Sample $Y^\ast \sim \MKer_\sigma(Y_0, \wc) $ \;
            Sample $ w^\ast \mid Y^\ast \sim \Unif([0, \MKer_\sigma(Y_0, Y^\ast) ]) $\;
        }
    }
    Sample $ u \sim \Unif([0, 1])$\;
    Set $ X = X_0 $ and $ Y = Y_0$ \;
    \If{$u \leq \min\{ 1, \pi(X^\ast) / \pi(X_0) \}$}{Set $X = X^\ast$}
    \If{$u \leq \min\{ 1, \pi(Y^\ast) / \pi(Y_0) \}$}{Set $Y = Y^\ast$}
    Output $(X, Y)$\;
    \caption{Coupled RWMH kernel with maximal coupling \citep{jacob_unbiased_2019}}
    \label{alg:coupled-mh}
\end{algorithm}

\section{Additional Algorithmic Details}
\label{app:method}

\setcounter{subsection}{1}
\subsection{Sampling from discrete joints}
For completeness, we provide an algorithmic description of how to sample a pair of indices given their joint probability matrix in Algorithm~\ref{alg:joint_sampling}.
\begin{algorithm}[t]
    \KwInput{A $M \times N$ matrix $J$ that represents the joint of two categorical distributions}
    \KwOutput{A pair of indices $(i,j) \sim J$}
    \For{$i=1,\dots,M, j=1,\dots,N$}
    {
        Compute $ k = M (i - 1) + j $\;
        Set $ u_k = (i, j) $ and $ \vv_k = J_{ij} $\;
    }
    Sample $ k \sim \Cat(\vv)$\;
    Output $ u_k $\;
    \caption{Sampling from a discrete joint $J$}
    \label{alg:joint_sampling}
\end{algorithm}

\subsection{Debiasing marginal-non-preserving joints}

A side effect of using fixed-point iteration solvers or even approximate solvers \citep{cuturi2013sinkhorn} to solve~\eqref{eq:ot_objective} is that
the solution does not belong to $ \Gamma(\vmu, \vnu) $.
We denote such solutions as $J^\circ$, 
which indicates it is a \underline{j}oint probability matrix rather than a proper coupling.
Therefore we need a way to ensure that when using $J^\circ$,
we still have $ i \sim \vmu $ and $ j \sim \vnu $ exactly,
which we refer as a \textit{debiasing} step.
Inspired by the mixture view of the maximal coupling,
the result of our debiasing algorithm, the debiased $W_2$-coupling $\hat{\gamma}^\circ$,
can be as well viewed as a mixture
$$
\hat{\gamma}^\circ = \alpha J^\circ + (1 - \alpha) J^d
$$
where $\alpha$ is the probability of sampling from $J^\circ$, 
and $J^d$ is the \underline{d}ebiasing joint probability matrix.
The algorithm aims to find the maximal probability $ \alpha $ such that $ \hat{\gamma}^\circ \in \Gamma(\vmu, \vnu) $,
together with the corresponding debiasing matrix $ J^d $.
First, to find the maximal $ \alpha $,
we see that $ \hat{\gamma}^\circ \in \Gamma(\vmu, \vnu) $ implies
\begin{equation}
\begin{aligned}
&\vmu = \alpha \vmu^\circ + (1 - \alpha) \vmu^d, &
&\vnu = \alpha \vnu^\circ + (1 - \alpha) \vnu^d
\end{aligned}
\label{eq:debiasing-marginal_constraints}
\end{equation}
where $\vmu^\circ$ and $\vnu^\circ$ are marginals of $J^\circ$ and
$\vmu^d$ and $\vnu^d$ are marginals of $J^d$.
Since $\vmu^d$ and $\vnu^d$ are $K$-length probability vectors, 
we have $\mu_i^d >0 $ and $\nu_i^d > 0$ for all $i = 1,\dots,K$, 
which implies a set of constrains on $\alpha$
\begin{equation*}
\begin{aligned}
&\mu_i \geq \alpha \mu^\circ_i,&
&\text{ and }&
&\nu_i \geq \alpha \nu^\circ_i&
&\text{for all}\;i  = 1, \dots, K
\end{aligned}
\end{equation*}
Therefore, the maximal value of $\alpha$ is given by
\begin{equation}
\alpha = \min\{1,  \frac{\mu_1}{\mu^\circ_1}, \dots, \frac{\mu_K}{\mu^\circ_K}, \frac{\nu_1}{ \nu^\circ_1}, \dots, \frac{\nu_K}{\nu^\circ_K} \}.
\label{eq:debiasing-probability}    
\end{equation}
With $ \alpha $ found, 
we can solve~\eqref{eq:debiasing-marginal_constraints} to find $\vmu^d$ and $\vnu^d$,
and $J^d$ can be chosen as any coupling of them, i.e. $ J^d \in \Gamma(\vmu^d, \vnu^d)$,
including the \textit{independent coupling} that simply samples as $ i \sim \vmu^d, j \sim \vnu^d $.
We summarise in \algref{alg:debiasing} a sampling procedure of $\hat{\gamma}^\circ$ resulting from 
this debiasing approach.
\begin{algorithm}[t]
\textsc{Input}: A $K \times K$ probability matrix $\hat{\gamma}$ and two $K$-length probability vectors $\vmu, \vnu$ to target \;
\textsc{Output}: A pair of indices $(i,j)$ with $i\sim\vmu$ and $j\sim\vnu$ while maximally using $\hat{\gamma}$ \;
Compute $\vmu^\circ$ and $\vnu^\circ$ as marginals of $\hat{\gamma}$\;
Compute $\alpha$ according to~\eqref{eq:debiasing-probability}\;
Sample $ U \sim \Unif([0, 1])$\;
\eIf{$U<\alpha$}{
    Sample $(i,j) \sim \hat{\gamma} $ using Algorithm~\ref{alg:joint_sampling}\;
}{
    Compute $\vmu^d$ and $ \vnu^d$ by solving~\eqref{eq:debiasing-marginal_constraints}\;
    Sample $i \sim \vmu^d$ and $j \sim \vnu^d$\;
}
Output $ (i, j) $\;
\caption{Maximally sampling from a joint $\hat{\gamma}$ while ensuring marginals to be $\vmu$ and $\vnu$}
\label{alg:debiasing}
\end{algorithm}
It is not hard to see that by construction,
the approach satisfies~\eqref{eq:debiasing-marginal_constraints} and 
yields $ \hat{\gamma}^\circ \in \Gamma(\vmu, \vnu) $,
which, as a result, yields a coupled HMC kernel whose marginal kernels converge to the target.
Also, when there is no bias, i.e. $J^\circ \in \Gamma(\vmu, \vnu)$, we have $\alpha = 1$ from~\eqref{eq:debiasing-probability} and the algorithm reduces to exact $W_2$-coupling.

\subsection{Sampling from discrete maximal maximal coupling}
For completeness, we provide an algorithmic description of how to sample a pair if indices from the maximal coupling of two categorical distribution in Algorithm~\ref{alg:maximal_coupling}.
\begin{algorithm}[t]
\textsc{Input}: Two categorical distributions $\vmu$ and $\vnu$\;
\textsc{Output}: A pair of indices $(i,j) \sim \gamma^\ast$ \;
Compute $\omega = 1 - \TV(\vmu, \vnu)$ and $Z = \sum_i (\vmu \wedge \vnu)_i$\;
Sample $ u \sim \Unif([0, 1])$\;
\eIf{$u \leq \omega$}
{
    Sample $i \sim \Cat(\frac{\vmu \wedge \vnu}{Z})$ and set $j=i$\;
}{
    Sample $i \sim \Cat(\frac{\vmu - (\vmu \wedge \vnu)}{1 - Z})$, $j \sim \Cat(\frac{\vnu - (\vmu \wedge \vnu)}{1 - Z})$\;
}
Output $(i, j)$\;
\caption{Maximal coupling of $\vmu$ and $\vnu$}
\label{alg:maximal_coupling}
\end{algorithm}

\section{Technical Details}
\label{app:proof}

\subsection{Proof of \lemref{lemma:local-contractivity-implies-prop-1}}
\label{app:proof-local-contractivity-implies-prop-1}

\begin{prf}
  Suppose $\KerCHMC_{\stepsize{}, L}^{\gamma}$ satisfies \condref{ass:local_contractivity} on the set $S$ for some $\bar{\stepsize{}} > 0$, $\bar{L} \in \mathbb{N}$.

  First observe that
  \begin{equation*}
    \begin{split}
      & \PkerCHMC{\stepsize{}}{L}{\gamma} \left( \norm{Q_1^1 - Q_1^2} \le \rho \norm{Q_0^1 - Q_0^2} \mid (Q_0^1, Q_0^2) = (q^1, q^2) \right) \\
      &= \E[\KerCHMC_{\stepsize{}, L}^{\gamma}]{\ind \set{\norm{Q_1^1 - Q_1^2} \le \rho \norm{q^1 - q^2}}} \\
      &= \mathbb{E}_{P \sim \mathcal{N}(0, I)} \Big[ \E[(l_1, l_2) \sim \gamma]{\ind(R_{q^1, q^2, P}) \mid P} \Big]
    \end{split}
  \end{equation*}
  where we have let $R_{q^1, q^2, p}$ denote the set of events where we have contraction,~i.e.
  \begin{equation*}
    R_{q^1, q^2, p} = \set{\norm{\flowqdisc{\stepsize{}}{l_1}(q^1, p) - \flowqdisc{\stepsize{}}{l_2}(q^2, p)} \le \rho \norm{q^1 - q^2}}
  \end{equation*}
  By \condref{ass:local_contractivity} we know that there exists $\omega_1 \in (0, 1)$ such that
  \begin{equation}\label{eq:gamma-prob-lower-bound}
    \Prob[(l_1, l_2) \sim \gamma]{R_{q^1, q^2, p}} \ge \omega_1
  \end{equation}
  for all $(q^1, q^2, p) \in S \times S \times \levelset{k_0}{K}$, where $k_0 > 0$.
  By the tower property of expectation, this immediately implies that
  \begin{equation*}
    \begin{split}
      & \mathbb{E}_{P \sim \mathcal{N}(0, I)} \Big[ \E[(l_1, l_2) \sim \gamma]{\ind(R_{q^1, q^2}) \ind \left\{ K(P) \le k_0 \right\} \mid P} \Big] \\
      &\ge \E[P \sim \mathcal{N}(0, I)]{\omega_1 \ind \left\{ K(P) \le k_0 \right\}} \\
      &= \omega_1 \Prob[P \sim \mathcal{N}(0, I)]{\set{K(P) \le k_0}} \\
      &> 0
    \end{split}
  \end{equation*}
  where the last inequality follows from the fact that the level sets $\levelset{k_0}{K}$ are closed for any $k_0 > 0$ since $K$ is continuous and bounded and therefore compact, in addition to having positive Lebesgue measure.
  Since~\eqref{eq:gamma-prob-lower-bound} holds for all $(q^1, q^2, p) \in S \times S \times \levelset{k_0}{K}$ with $\omega_1 > 0$, we have
  \begin{equation}\label{eq:contraction-in-probability-infimum-with-K}
  \begin{split}
    \inf_{q^1, q^2 \in S} \PkerCHMC{\stepsize{}}{L}{\gamma} \big( & \left\{ \norm{Q_1^1 - Q_1^2} \le \rho \norm{Q_0^1 - Q_0^2} \right\} \\
    & \cap \set{K(P) \le k_0} \mid (Q_0^1, Q_0^2) = (q^1, q^2) \big) \\
    & \ge \omega_1 \omega_2 \\
    & > 0
  \end{split}
  \end{equation}
  where we have let $\omega_2 = \Prob[P \sim \mathcal{N}(0, I)]{K(P) \le k_0}$.
  
  In words, for any initial points $(q^1, q^2) \in S \times S$, a single application of the kernel $\KerCHMC_{\stepsize{}, L}^{\gamma}$ decreases the distance with non-zero probability. Equipped with this, proving the desired statement is just a matter of ensuring that we can indeed apply \eqref{eq:contraction-in-probability-infimum-with-K} repeatedly to get the states sufficiently close to each other. A straightforward approach to this is to simply choose the stepsize to be sufficiently small such that even when taking the required number of steps to get within the desired $\delta \text{-ball}$, every step taken is still within a set where \eqref{eq:contraction-in-probability-infimum-with-K} holds. This is exactly the approach taken in \citet{heng_unbiased_2019} and so the rest of the proof is essentially identical to the last paragraph in the proof of Proposition 1 in \citet{heng_unbiased_2019}.
  
  Consider $u_0 > \inf_{q \in S} U(q)$, and $u_1 < \sup_{q \in S} U(q)$ with $u_0 < u_1$, and let $A_{\ell} := \levelset{\ell}{U_S} \times \levelset{u_1 - \ell}{K} \subset \levelset{u_1}{\Energy}$ for $\ell \in (u_0, u_1)$. Since continuity and convexity of $U_S$ imply that this is a closed function, its level sets $\levelset{\ell}{U_S}$ are closed. 
  Moreover, under the assumptions on $U$ and $S$, it follows that these level sets are compact with positive Lebesgue measure.
  Note that if $(q, p) \in A_{\ell}$, due to energy conservation and continuity of $U$, the mapping $t \mapsto \flowq{t}(q, p)$ imply that $\flowq{t}(q, p) \in \levelset{u_1}{U_S}$ for any $t \in [-T, T]$.
  Due to time discretization, using~\eqref{eq:numerical-error-energy} and compactness of $A_{\ell}$ we can only conclude that there exists $\eta_0 > 0$ such that $\flowqdisc{\stepsize{}}{l}(q, p) \in \levelset{u_1 + \eta_0}{U}$ for all $(q, p) \in A_{\ell}$ and $l = L_{b}, ..., L_{f}$.
  Let $n_0 = \min \set{n \in \mathbb{N}: \rho^n B \le \delta}$, where $B := \sup_{q^1, q^2 \in S} \norm{q^1 - q^2}$.
  By choosing $v_0 \in (u_0, u_1)$, $k_0 > 0$, and $\eta_0 > 0$ small enough such that
  \begin{equation*}
    v_0 + (n_0 + 1) k_0 + n_0 \eta_0 < u_1
  \end{equation*}
  holds, we have $Q_k^1, Q_k^2 \in S$ for all $k = 1, \dots, n_0$.
  Hence, by repeated application of~\eqref{eq:contraction-in-probability-infimum-with-K},
  \begin{equation*}
    \inf_{q^1, q^2 \in S_0} \PkerHMC{\stepsize{}}{L} \left( \norm{ Q_{n_0}^1 - Q_{n_0}^2 } \le \delta \mid (Q_0^1, Q_0^2) = (q^1, q^2) \right) > 0
  \end{equation*}
  with $S_0 = \levelset{v_0}{U_S}$, exactly as in \propref{prop:proposition-1-heng2019}.
\end{prf}

\subsection{Proof of Lemma~\ref{lemma:local_abs_contractivity}}
\label{app:proof-local_abs_contractivity}

\begin{lemma}\label{lemma:local_abs_contractivity_disc}
    Suppose that the potential $U$ satisfies Assumptions \ref{ass:regularity_and_growth} and \ref{ass:local_strong_convexity}. For any compact set $A \subset S \times S \times \mathbb{R}^d$, there exists a trajectory length $T > 0$ and a step size $\stepsize{}_1 > 0 $ s.t.~for any $\stepsize{} \in (0, \stepsize{}_1]$ and any $t \in [-T, T] \setminus \set{0}$ with $l := t / \stepsize{} \in \ZZ$,
    there exists $\rho \in [0, 1)$ satisfying
    \begin{equation}
        \norm{\flowqdisc{\stepsize{}}{l}(q_0^1, p_0) - \flowqdisc{\stepsize{}}{l}(q_0^2, p_0)} \leq \rho \norm{q_0^1 - q_0^2}
        \label{eq:local_abs_contractivity_disc}
    \end{equation}
    for all $(q_0^1, q_0^2, p_0) \in A$.
\end{lemma}

\begin{proof}
    As the leapfrog integrator is of order two \citep{hairer2006geometric,bou-rabee_coupling_2019}, for any sufficently small step size $\stepsize{}$ and number of step $l$ states above, we have
    $ \norm{\flowdisc{\stepsize{}}{l}(q_0, p_0) - \flow{t}(q_0, p_0)} \leq C_1(q_0, p_0, t) \stepsize{}^2 $ and similar for its position-projected correspondence
    \begin{equation}
        \norm{\flowqdisc{\stepsize{}}{l}(q_0, p_0) - \flowq{t}(q_0, p_0)} \leq C_1(q_0, p_0, t) \stepsize{}^2
        \label{eq:leapfrog_bound-phase}
    \end{equation}
    where $C_1(q_0, p_0, t)$ is some constant that only depends on $q_0, p_0$ and $t$.

    By \citep[Lemma 1,][]{heng_unbiased_2019}, with some fixed $T$,
    we have $\rho' \in [0, 1)$ satisfying
    \begin{equation}
        \norm{\flowq{t}(q_0^1, p_0) - \flowq{t}(q_0^2, p_0)} \leq \rho' \norm{q_0^1 - q_0^2}
    \label{eq:heng_and_jacob-lemma1}
    \end{equation}
    for any $t \in (0, T]$ and all $(q_0^1, q_0^2, p_0) \in A$.
    Since $\flowq{t}(q_0^1, -p_0) = \flowq{-t}(q_0^1, p_0)$,
    applying \citep[Lemma 1,][]{heng_unbiased_2019} again with the momentum variable negated,
    we have~\eqref{eq:heng_and_jacob-lemma1} for $t \in [-T, 0)$.
    Therefore~\eqref{eq:heng_and_jacob-lemma1} holds for $t \in [-T, T] \setminus \set{0}$.

    With these two intermediate results, 
    we can now bound the left-hand side (LHS) of~\eqref{eq:local_abs_contractivity} for any $t \in [-T, T] \setminus \set{0}$ with $l = t / \stepsize{} \in \ZZ$ and all $(q_0^1, q_0^2, p_0) \in A$
    \begin{equation*}
    \begin{aligned}
        &\norm{{\color{red} \flowqdisc{\stepsize{}}{l}(q_0^1, p_0)} - \color{blue} {\flowqdisc{\stepsize{}}{l}(q_0^2, p_0)}} \\
        =& \lVert{\color{red} \flowqdisc{\stepsize{}}{l}(q_0^1, p_0)} - {\color{orange} \flowq{t}(q_0^1, p_0)} - \\&{\color{blue} \flowqdisc{\stepsize{}}{l}(q_0^2, p_0)} + {\color{brown} \flowq{t}(q_0^2, p_0)} + {\color{orange} \flowq{t}(q_0^1, p_0)} - {\color{brown} {\flowq{t}(q_0^2, p_0)}} \rVert \\
        \leq& \norm{{\color{red} \flowqdisc{\stepsize{}}{l}(q_0^1, p_0)} - {\color{orange} \flowq{t}(q_0^1, p_0)}} + \\&\norm{{\color{blue} \flowqdisc{\stepsize{}}{l}(q_0^2, p_0)} - {\color{brown} \flowq{t}(q_0^2, p_0)}} + \norm{{\color{orange} \flowq{t}(q_0^1, p_0)} - \color{brown} {\flowq{t}(q_0^2, p_0)}} \\
        \leq& \left( C(q_0^1, p_0, t) + C(q_0^2, p_0, t) \right) \stepsize{}^2 + \rho' \norm{q_0^1 - q_0^2}
    \end{aligned}
    \end{equation*}
    where the third line is a result of the triangle inequality and the last line comes from~\eqref{eq:leapfrog_bound-phase} and~\eqref{eq:heng_and_jacob-lemma1} respectively.
    As $\lim_{\stepsize{} \to 0} \left( C(q_0^1, p_0, t) + C(q_0^2, p_0, t) \right) \stepsize{}^2 = 0$,
    for any $\rho \in (\rho', 1)$,
    there exists a step size $\stepsize{}_1 > 0$ such that for any $\stepsize{} \leq \stepsize{}_1$,~\eqref{eq:local_abs_contractivity} holds.
\end{proof}

\subsection{Proof of Proposition~\ref{prop:property-parallel}}
\label{app:proof-property-parallel}

For the sake of presentation, in this section we only consider \condref{ass:local_contractivity} for $m = 1$. To prove that $\gamma^{\ast}$ satisfies \condref{ass:local_contractivity} for $m > 1$ follows the exact reasoning since \eqref{eq:local_abs_contractivity} in \lemref{lemma:local_abs_contractivity} still holds when both sides are raised to some positive power $m$.

\begin{proposition}\label{prop:property-parallel}
    Suppose that $U$ satisfies Assumptions \ref{ass:regularity_and_growth} and \ref{ass:local_strong_convexity}.
    For any compact set $A \subset S \times S \times \mathbb{R}^d$ and any parallel-in-time joint $J^\parallel \in \RR^{K \times K}$, 
    there exists a trajectory length $T > 0$, a step size $\stepsize{}_1 > 0 $ s.t.~for any $\stepsize{} \in (0, \stepsize{}_1]$ and any $L_1, L_2 \in \NN$ with $L_1 + L_2 = K - 1$ and $\stepsize{} L_1, \stepsize{} L_2 < T$,
    there exists $\tilde{\rho} \in (0, 1)$ satisfying
    \begin{equation}
        \E[(i, j) \sim J^\parallel] { \norm{\flowqdisc{\stepsize{}}{l_i}(q^1, p) - \flowqdisc{\stepsize{}}{l_j}(q^2, p)} } \leq \tilde{\rho} \norm{q^1 - q^2}
        \label{eq:property-parallel}
    \end{equation}
    for all $(q^1, q^2, p) \in A$,
    where $l_k$ is the $k$-th entry of the vector $[-L_1, \dots, 0, \dots, L_2]$.
\end{proposition}

\begin{proof}
    By definition, $J^\parallel$ has only diagonal entries, thus $(i,j) \sim J^\parallel$ is equivalent to $(i, i)$ with $ i \sim \mathrm{diag}(J^\parallel)$.
    Denote the left-hand side of~\eqref{eq:property-parallel} as $A_1$,
    expanding and rearranging $A_1$ and applying \lemref{lemma:local_abs_contractivity}, we have
    \begin{equation*}
    \begin{aligned}
        A_1
        &= \sum_{k = 0}^{L_1+L_2+1} \PP(i = k) \times \norm{\flowqdisc{\stepsize{}}{l_k}(q^1, p) - \flowqdisc{\stepsize{}}{l_k}(q^2, p)} \\
        &= \sum_{k \neq L_1+1} \PP(i = k) \times \norm{\flowqdisc{\stepsize{}}{l_k}(q^1, p) - \flowqdisc{\stepsize{}}{l_k}(q^2, p)} \\
        & \qquad \qquad + \PP(i = L_1 + 1) \times \norm{\flowqdisc{\stepsize{}}{0}(q^1, p) - \flowqdisc{\stepsize{}}{0}(q^2, p)} \\
        &\leq \sum_{k \neq L_1+1} \PP(i = k) \times \rho_{l_k} \times \norm{q^1 - q^2} \\
        & \qquad \qquad + \PP(i = L_1 + 1) \times \norm{q^1 - q^2} \\
        &= \E[i]{ \rho_{l_i}} \times \norm{q^1 - q^2} \\
        &:= \tilde{\rho} \norm{q^1 - q^2}
    \end{aligned}
    \end{equation*}
    where we let $\rho_0 = 1$.
    As $\rho_l \in (0, 1) $ for $l \neq 0$ and $\rho_0 = 1$,
    $ \E[i]{ \rho_{l_i}} = \sum_k \PP(i = k) \times \rho_{l_k} \in (0, 1) $ by the property of convex combination.
    In other words, we have $\tilde{\rho} \in (0, 1)$.
\end{proof}

\subsection{Proof of Proposition~\ref{prop:maximal_coupling}}
\label{app:proof-maximal_coupling}

\begin{proof}
    For two length-$K$ Hamiltonian trajectories $\traj^1$ and $\traj^2$,
    denote $\vx = [\Energy(\traj^1_1), \dots, \Energy(\traj^1_K)]$ and $\vy = [\Energy(\traj^2_1), \dots, \Energy(\traj^2_K)]$ as vectors of the Hamiltonian energy of all phasepoints.
    With the softmax function $\sigma(\vx)_i = \exp(-\vx_i) / \sum_{i'} \exp(-\vx_{i'})$,
    the entries of $\vmu$ and $\vnu$ can be expressed as
    \begin{equation*}
        \begin{aligned}
        &\mu_i = \sigma(\vx)_i&
        &\nu_j = \sigma(\vy)_j
        \end{aligned}
    \end{equation*}
    By the Cauchy–Schwarz inequality, we have $\norm{\sigma(\vx) - \sigma(\vy)}_1 \leq \sqrt{K} \norm{\sigma(\vx) - \sigma(\vy)}$.
    With this, we can then upper-bound $\TV( \vmu, \vnu )$ as
    \begin{equation*}
    \begin{aligned}
        \TV( \vmu, \vnu ) &= \TV( \sigma(\vx), \sigma(\vy) ) \\
        &= \frac{1}{2}\norm{\sigma(\vx) - \sigma(\vy)}_1 \\
        &\leq \frac{1}{2} \sqrt{K} \norm{\sigma(\vx) - \sigma(\vy)}
    \end{aligned}
  \end{equation*}
    Denote the energy of the initial phasepoints in each trajectory $(q_0^1, p_0)$ and $(q_0^2, p_0)$ as $\Energy_0^1$ and $\Energy_0^1$ and let $\Energy_i^1 := \vx_i$ and $\Energy_j^1 := \vy_j$; note that for some $i_0 \in \set{1, \dots, K}$ we have $\Energy(\traj_{i_0}^c) = \Energy_0^c$ for $c = 1, 2$, i.e. $i_0$ represents the initial time-index which is shared between the two.
    As the leapfrog integrator is of order two \citep{hairer2006geometric,bou-rabee_coupling_2019}, for any sufficiently small step size $\stepsize{} = T / L$, we have
    \begin{equation}\label{eq:energy-numerical-error}
        \abs{\Energy_0^c - \Energy(\traj_i^c))} \leq C_2(q_0^c, p_0)\,t_i\,\stepsize{}^2 \leq C_2(q_0^c, p_0)\,T\,\stepsize{}^2
    \end{equation}
    for $c = 1, 2$, where $t_i$ denotes the corresponding integration time for the $i$-th phasepoint from the first phasepoint.
    Denote the energy differences as $\Delta_i^1 = \Energy(\traj_i^1) - \Energy_0^1$ and $\Delta_j^1 = \Energy(\traj_j^2) - \Energy_0^2$ and observe that
    \begin{equation*}
        \begin{aligned}
            &\sigma(\vx) = \sigma([\Delta_1^1, \dots, \Delta_K^1])&
            &\sigma(\vy) = \sigma([\Delta_1^2, \dots, \Delta_K^2])
        \end{aligned}
    \end{equation*}
    Using the fact that the softmax function is 1-Lipschitz \citep{gao2018properties} and applying~\eqref{eq:energy-numerical-error}, we have
    \begin{equation*}
        \begin{aligned}
            \norm{\sigma(\vx) - \sigma(\vy)}
            &= \norm{\sigma([\Delta_1^1, \dots, \Delta_K^1]) - \sigma([\Delta_1^2, \dots, \Delta_K^2])} \\
            &\le \norm{[\Delta_1^1, \dots, \Delta_K^1] - [\Delta_1^2, \dots, \Delta_K^2]} \\
            &\leq \sqrt{ \sum_{k=1}^K C_2(q_0^1, p_0) C_2(q_0^2, p_0)\,T^2\,\stepsize{}^4 } \\
            &= \sqrt{K C_2(q_0^1, p_0) C_2(q_0^2, p_0)}\,T\,\stepsize{}^2
        \end{aligned}
    \end{equation*}
    Substituting back into our bound on $\TV( \vmu, \vnu )$,
    \begin{equation*}
        \TV( \vmu, \vnu ) \leq \frac{1}{2} K \sqrt{C_2(q_0^1, p_0) C_2(q_0^2, p_0)}\,T\,\stepsize{}^2
    \end{equation*}
    Since $T$ is fixed, $\stepsize{} = T / L$ and $K = L + 1$, we have
    \begin{equation}
    \begin{split}
        \TV( \vmu, \vnu ) &\leq \frac{1}{2} \sqrt{C_2(q_0^1, p_0) C_2(q_0^2, p_0)}\,T^3\,\frac{L+1}{L^2}  \\
        & \le \sqrt{C_2(q_0^1, p_0) C_2(q_0^2, p_0)}\,T^3\,L^{-1} \\
        & \le \sqrt{C_2(q_0^1, p_0) C_2(q_0^2, p_0)}\,T^2\,\stepsize{}.
    \end{split}
    \label{eq:tv_bound}
    \end{equation}
    Finally, note that the upper-bound decreases in with $\stepsize{}$ and $T$, hence for any given $\delta > 0$, there exists $\stepsize{}_0 > 0$, $L_0 \in \mathbb{N}$ such that $\TV( \vmu, \vnu ) \le \delta$ for all $\stepsize{} \in (0, \stepsize{}_0)$ and $L \in \mathbb{N}$ satisfying $\stepsize{} L < \stepsize{}_0 L_0 = T$.
\end{proof}

\subsection{Proof of Lemma~\ref{lemma:contractivity-maximal}}
\label{app:proof-contractivity-maximal}

Similarly to in \appref{app:proof-property-parallel} we only consider \condref{ass:local_contractivity} with $m = 1$ as the case of $m > 1$ follows similarly.

To prove \lemref{lemma:contractivity-maximal} we first restate a more detailed version of the lemma, which we then prove.

\begin{lemma}
    Suppose that the potential $U$ satisfies Assumptions \ref{ass:regularity_and_growth} and \ref{ass:local_strong_convexity}.
    For a maximal coupling $\gamma^\ast$, 
    there exists a trajectory length $T > 0$ and a step size $\stepsize{}_2 > 0 $ such that for any $\stepsize{} \in (0, \min\{\stepsize{}_1, \stepsize{}_2\}]$ and any $t \in [-T, T] \setminus \set{0}$ with $l := t / \stepsize{} \in \ZZ$,
    there exists $\rho_2 \in (0, 1)$ satisfying
    \begin{equation}\label{eq:contractivity-maximal}
        \E[(l_1, l_2) \sim \gamma^{\ast}]{\norm{\flowqdisc{\stepsize{}}{l_1}(q^1, p) - \flowqdisc{\stepsize{}}{l_2}(q^2, p)}} \le \rho \norm{q^1 - q^2}
    \end{equation}
    for all $(q^1, q^2) \in S \times S$,
    where $\CMKer_{\stepsize, l}^\ast$ is the coupled kernel in Algorithm~\ref{alg:coupled-hmc} with (i) shared momentum, (ii) shared forward and backward simulation steps and (iii) $(i, j) \sim \gamma^\ast$ for intra-trajectory sampling.
    \label{lemma:contractivity-maximal-restated}
\end{lemma}

\begin{proof}
    We first decompose $\gamma^\ast$ into its "diagonal" and "non-diagonal" components
    \begin{equation*}
        \gamma^\ast = \omega J^\parallel + (1 - \omega) J^\nparallel
    \end{equation*}
    where $1 - \omega = \PP(i \neq j)$ and $J^\nparallel$ is defined to be the residual with normalization.
    Thus we have
    \begin{equation}\label{eq:lemma-4.2-A2}
    \begin{aligned}
    & A_2 := \omega \E[J^\parallel] { \norm{\flowqdisc{\stepsize{}}{l_i}(q^1, p) - \flowqdisc{\stepsize{}}{l_j}(q^2, p)} } \\
    &\quad + (1 - \omega) \E[J^\nparallel] { \norm{\flowqdisc{\stepsize{}}{l_i}(q^1, p) - \flowqdisc{\stepsize{}}{l_j}(q^2, p)} } \\
    & \leq \omega \tilde{\rho} \norm{q^1 - q^2} \\
    & \quad + (1 - \omega) \E[J^\nparallel] { \norm{\flowqdisc{\stepsize{}}{l_i}(q^1, p) - \flowqdisc{\stepsize{}}{l_j}(q^2, p)} } 
    \end{aligned}
    \end{equation}
    for $T > 0$, $\stepsize{} \in (0, \stepsize{}_1]$ and $\tilde{\rho} \in (0, 1) $ in Proposition~\ref{prop:property-parallel}.
    As $\E[J^\nparallel] { \norm{\flowqdisc{\stepsize{}}{l_i}(q^1, p) - \flowqdisc{\stepsize{}}{l_j}(q^2, p)}}$ is finite,
    by Proposition~\ref{prop:maximal_coupling},
    the limit of the upper bound goes to  $\tilde{\rho} \norm{q^1 - q^2}$ as $\stepsize{} \to 0$.
    In other words, 
    for any $\rho \in (\tilde{\rho}, 1)$,
    there exists a step size $\stepsize{}_2 > 0$ such that for any $\stepsize{} \in (0, \min\{\stepsize{}_1, \stepsize{}_2\}]$,
    \begin{equation*}
        A_2 \leq \rho \norm{q^1 - q^2}
    \end{equation*}
    which is exactly what we wanted to prove.
\end{proof}

\section{Additional Experimental Details}

\subsection{Target distributions}
\label{app:target_distributions}

We follow the pre-processing steps in \cite{heng_unbiased_2019} for the German credit dataset \citep{asuncion2007uci} and the Finnish pine saplings dataset \citep{moller_log_1998} used in logistic regression and log-Gaussian Cox point process respectively.

\paragraph{Bayesian logistic regression}
We combine features in the German credit dataset  with all of their standardized pairwise interactions, resulting in a design matrix in $\RR^{300 \times 1,000}$. Denoting an Exponential distribution with rate $\lambda$ as $\Exp(\lambda)$, the Bayesian logistic regression follows the following generative process: $s^2 \sim \Exp(\lambda), a \sim \Normal(0, s^2), b \sim \Normal_{300}$, where the variance $s^2\in \RR$, the intercept $a \in \RR$ and the coefficients $b \in \RR^{300}$, giving a total dimension $d=302$.

\paragraph{Log-Gaussian Cox point process}

Firstly, the plot of the forest is discretized into an $n \times n$ grid.
For $i \in \{1, \dots, n\}^2$,
the number of points in each grid cell $y_i \in \mathbb{N}$ is assumed to be conditionally independent given a latent intensity variable $\Lambda_i$ and
follows a Poisson distribution with mean $a \Lambda_i$,
where $a = n^{-2}$ is the area of each cell.
We denote the logarithm of $\Lambda$ as $X$ and put a Gaussian process prior with mean $\mu \in \mathbb{R}$ and exponential covariance function $\Sigma_{i,j} = s^2 \exp \left(-|i - j| / (n b)\right)$ on it, where $s^2$, $b$ and $\mu$ are hyperparameters.
The generative process of the number of grid cell points follows
$X \sim \mathcal{GP}(\mu, \Sigma),\;\forall\;i \in \{1, \dots, n\}^2: \Lambda_i = \exp(X_i),\;y_i \sim \mathcal{P}oisson(a \Lambda_i)$.
Following \citep{moller_log_1998}, we use a dataset of 126 Scot pine saplings in a natural forest in Finland, 
and adapt the parameters $s^2 = 1.91$, $b = 1 / 33$ and $\mu = \log(126) - s^2 / 2$.

\section{Additional Experimental Results}

\subsection{Robustness: meeting time with more parameter sweeps}
\label{app:sweep}

Figure~\ref{app:meeting-gaussian}, \ref{app:meeting-lr} and \ref{app:meeting-coxprocess} 
provide a wider range of parameter sweep under the same experimental setup as Section~\ref{sec:meeting}.
\begin{figure*}[t]
    \ffigbox[\textwidth]{%
        \centering
        \begin{tikzpicture}
    \begin{axis}[%
        hide axis, xmin=10, xmax=50, ymin=0, ymax=0.4,
        legend style={color={rgb,1:red,0.1333;green,0.1333;blue,0.3333}, draw opacity={0.1}, line width={1}, solid, fill={rgb,1:red,1.0;green,1.0;blue,1.0}, fill opacity={0.9}, text opacity={1.0}, font={{\fontsize{8 pt}{10.4 pt}\selectfont}}, at={(1.02, 1)}, anchor={north west}, legend columns=-1}
    ]
        \addlegendimage{color={rgb,1:red,0.2667;green,0.4667;blue,0.6667}, name path={41d92112-fced-469f-bc61-94814c6a6591}, draw opacity={0.7}, line width={1.2}, solid, mark={diamond*}, mark size={2.25 pt}, mark options={color={rgb,1:red,0.0;green,0.0;blue,0.0}, draw opacity={0.7}, fill={rgb,1:red,0.2667;green,0.4667;blue,0.6667}, fill opacity={0.7}, line width={0.0}, rotate={0}, solid}}
        \addlegendentry{Metropolis}
        \addlegendimage{color={rgb,1:red,0.1333;green,0.5333;blue,0.2}, name path={d1056c54-bf38-4a84-9416-9b3b5e66f354}, draw opacity={0.7}, line width={1.2}, solid, mark={star}, mark size={2.25 pt}, mark options={color={rgb,1:red,0.0;green,0.0;blue,0.0}, draw opacity={0.7}, fill={rgb,1:red,0.1333;green,0.5333;blue,0.2}, fill opacity={0.7}, line width={0.0}, rotate={0}, solid}}
        \addlegendentry{Maximal}
        \addlegendimage{color={rgb,1:red,0.8;green,0.7333;blue,0.2667}, name path={98ad80d3-1b99-4400-911a-13702f6a1df5}, draw opacity={0.7}, line width={1.2}, solid, mark={*}, mark size={2.25 pt}, mark options={color={rgb,1:red,0.0;green,0.0;blue,0.0}, draw opacity={0.7}, fill={rgb,1:red,0.8;green,0.7333;blue,0.2667}, fill opacity={0.7}, line width={0.0}, rotate={0}, solid}}
        \addlegendentry{$W_2$}
    \end{axis}
\end{tikzpicture}
        \begin{subfloatrow}
            \ffigbox[0.32\textwidth]
            {\caption{$L = 5$}}
            {\includegraphics[width=\linewidth]{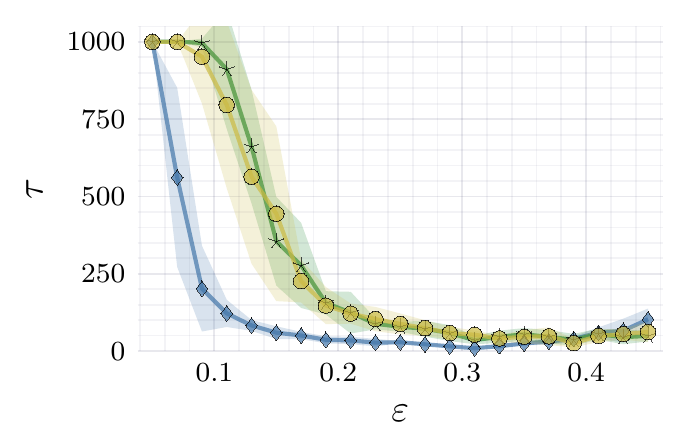}}
            \ffigbox[0.32\textwidth]
            {\caption{$L = 10$}}
            {\includegraphics[width=\linewidth]{plots/meeting/model=gaussian-L=10.pdf}}
            \ffigbox[0.32\textwidth]
            {\caption{$L = 15$}}
            {\includegraphics[width=\linewidth]{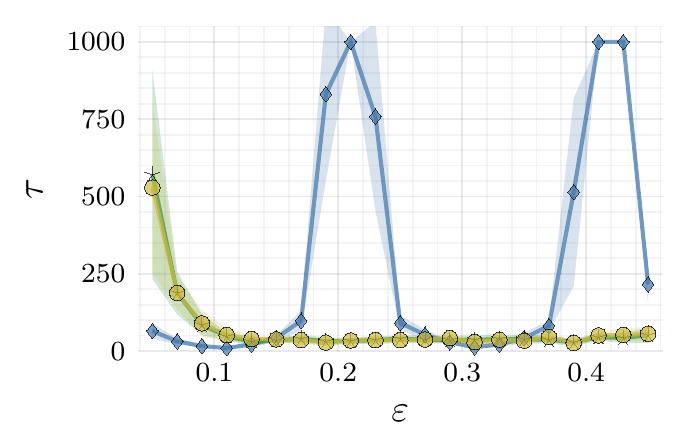}}
        \end{subfloatrow}
    }{%
        \caption{Averaged meeting time $\bar{\tau}$ with different $\epsilon$ and $L$ for 1,000D Gaussian.}
        \label{app:meeting-gaussian}
    }
\end{figure*}
\begin{figure*}[t]
    \ffigbox[\textwidth]{%
        \centering
        \begin{tikzpicture}
    \begin{axis}[%
        hide axis, xmin=10, xmax=50, ymin=0, ymax=0.4,
        legend style={color={rgb,1:red,0.1333;green,0.1333;blue,0.3333}, draw opacity={0.1}, line width={1}, solid, fill={rgb,1:red,1.0;green,1.0;blue,1.0}, fill opacity={0.9}, text opacity={1.0}, font={{\fontsize{8 pt}{10.4 pt}\selectfont}}, at={(1.02, 1)}, anchor={north west}, legend columns=-1}
    ]
        \addlegendimage{color={rgb,1:red,0.2667;green,0.4667;blue,0.6667}, name path={41d92112-fced-469f-bc61-94814c6a6591}, draw opacity={0.7}, line width={1.2}, solid, mark={diamond*}, mark size={2.25 pt}, mark options={color={rgb,1:red,0.0;green,0.0;blue,0.0}, draw opacity={0.7}, fill={rgb,1:red,0.2667;green,0.4667;blue,0.6667}, fill opacity={0.7}, line width={0.0}, rotate={0}, solid}}
        \addlegendentry{Metropolis}
        \addlegendimage{color={rgb,1:red,0.1333;green,0.5333;blue,0.2}, name path={d1056c54-bf38-4a84-9416-9b3b5e66f354}, draw opacity={0.7}, line width={1.2}, solid, mark={star}, mark size={2.25 pt}, mark options={color={rgb,1:red,0.0;green,0.0;blue,0.0}, draw opacity={0.7}, fill={rgb,1:red,0.1333;green,0.5333;blue,0.2}, fill opacity={0.7}, line width={0.0}, rotate={0}, solid}}
        \addlegendentry{Maximal}
        \addlegendimage{color={rgb,1:red,0.8;green,0.7333;blue,0.2667}, name path={98ad80d3-1b99-4400-911a-13702f6a1df5}, draw opacity={0.7}, line width={1.2}, solid, mark={*}, mark size={2.25 pt}, mark options={color={rgb,1:red,0.0;green,0.0;blue,0.0}, draw opacity={0.7}, fill={rgb,1:red,0.8;green,0.7333;blue,0.2667}, fill opacity={0.7}, line width={0.0}, rotate={0}, solid}}
        \addlegendentry{$W_2$}
    \end{axis}
\end{tikzpicture}
        \begin{subfloatrow}
            \ffigbox[0.32\textwidth]
            {\caption{$L = 10$}}
            {\includegraphics[width=\linewidth]{plots/meeting/model=lr-L=10.pdf}}
            \ffigbox[0.32\textwidth]
            {\caption{$L = 20$}}
            {\includegraphics[width=\linewidth]{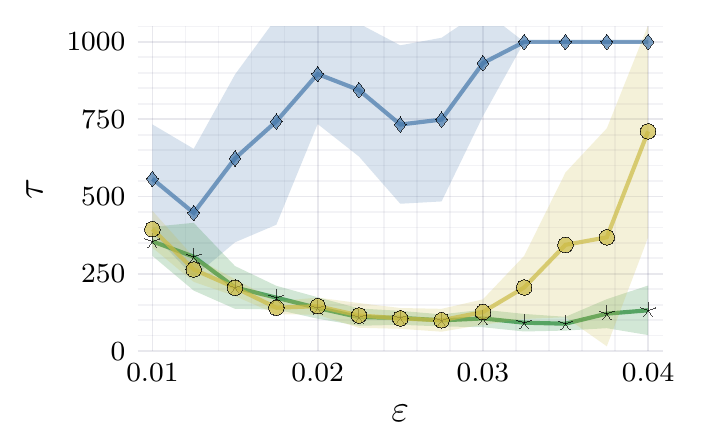}}
            \ffigbox[0.32\textwidth]
            {\caption{$L = 30$}}
            {\includegraphics[width=\linewidth]{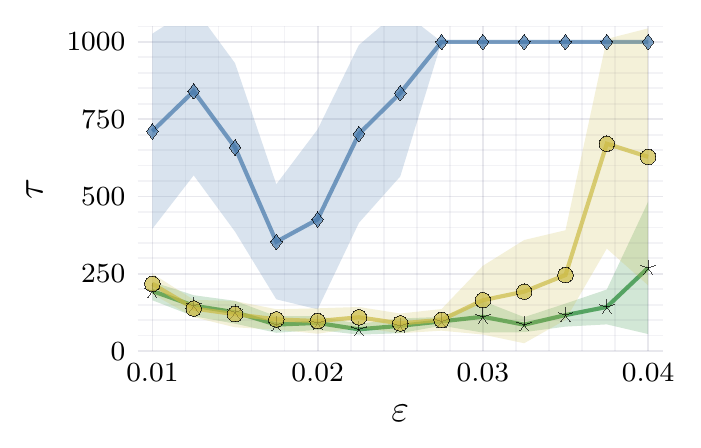}}
        \end{subfloatrow}
    }{%
        \caption{Averaged meeting time $\bar{\tau}$ with different $\epsilon$ and $L$ for logistic regression.}
        \label{app:meeting-lr}
    }
\end{figure*}
\begin{figure*}[t!]
    \ffigbox[\textwidth]{%
        \centering
        \begin{tikzpicture}
    \begin{axis}[%
        hide axis, xmin=10, xmax=50, ymin=0, ymax=0.4,
        legend style={color={rgb,1:red,0.1333;green,0.1333;blue,0.3333}, draw opacity={0.1}, line width={1}, solid, fill={rgb,1:red,1.0;green,1.0;blue,1.0}, fill opacity={0.9}, text opacity={1.0}, font={{\fontsize{8 pt}{10.4 pt}\selectfont}}, at={(1.02, 1)}, anchor={north west}, legend columns=-1}
    ]
        \addlegendimage{color={rgb,1:red,0.2667;green,0.4667;blue,0.6667}, name path={41d92112-fced-469f-bc61-94814c6a6591}, draw opacity={0.7}, line width={1.2}, solid, mark={diamond*}, mark size={2.25 pt}, mark options={color={rgb,1:red,0.0;green,0.0;blue,0.0}, draw opacity={0.7}, fill={rgb,1:red,0.2667;green,0.4667;blue,0.6667}, fill opacity={0.7}, line width={0.0}, rotate={0}, solid}}
        \addlegendentry{Metropolis}
        \addlegendimage{color={rgb,1:red,0.1333;green,0.5333;blue,0.2}, name path={d1056c54-bf38-4a84-9416-9b3b5e66f354}, draw opacity={0.7}, line width={1.2}, solid, mark={star}, mark size={2.25 pt}, mark options={color={rgb,1:red,0.0;green,0.0;blue,0.0}, draw opacity={0.7}, fill={rgb,1:red,0.1333;green,0.5333;blue,0.2}, fill opacity={0.7}, line width={0.0}, rotate={0}, solid}}
        \addlegendentry{Maximal}
        \addlegendimage{color={rgb,1:red,0.8;green,0.7333;blue,0.2667}, name path={98ad80d3-1b99-4400-911a-13702f6a1df5}, draw opacity={0.7}, line width={1.2}, solid, mark={*}, mark size={2.25 pt}, mark options={color={rgb,1:red,0.0;green,0.0;blue,0.0}, draw opacity={0.7}, fill={rgb,1:red,0.8;green,0.7333;blue,0.2667}, fill opacity={0.7}, line width={0.0}, rotate={0}, solid}}
        \addlegendentry{$W_2$}
    \end{axis}
\end{tikzpicture}
        \begin{subfloatrow}
            \ffigbox[0.32\textwidth]
            {\caption{$L = 10$}}
            {\includegraphics[width=\linewidth]{plots/meeting/model=coxprocess-L=10.pdf}}
            \ffigbox[0.32\textwidth]
            {\caption{$L = 20$}}
            {\includegraphics[width=\linewidth]{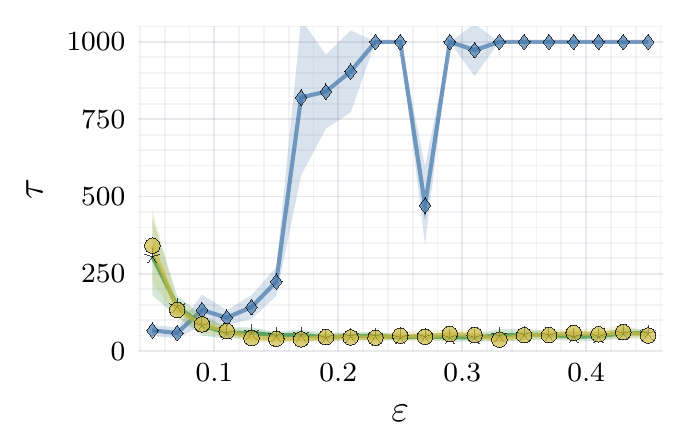}}
            \ffigbox[0.32\textwidth]
            {\caption{$L = 30$}}
            {\includegraphics[width=\linewidth]{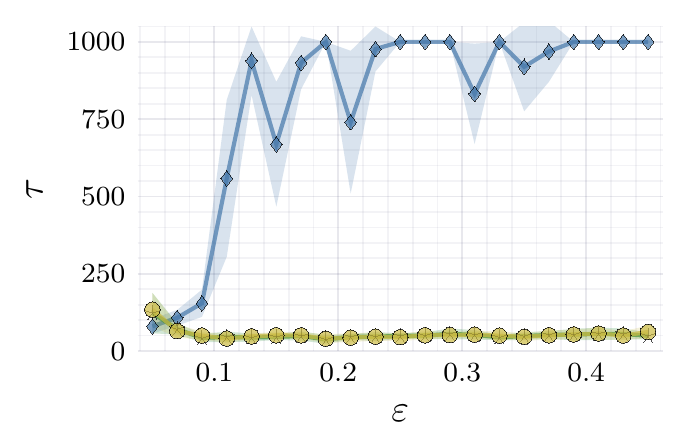}}
        \end{subfloatrow}
    }{%
        \caption{Averaged meeting time $\bar{\tau}$ with different $\epsilon$ and $L$ for log-Gaussian Cox point process.}
        \label{app:meeting-coxprocess}
    }
\end{figure*}

\subsection{Toy examples}
\label{app:toy}

We first study how proposed methods behave on multi-modal distributions.
Specifically, we want to know if the coupled chains can meet in a short time given the target is multi-modal.
We consider a mixture of Gaussians on $\mathbb{R}^2$ with three components $\Normal([-1, -1], 0.25^2 I)$, $\Normal([0, 0], 0.25^2 I)$, $\Normal([1, 1], 0.25^2 I)$ weighted by 0.25, 0.4 and 0.35 respectively.
We initialise chains from $\Unif([0, 1]^2)$, covering two of the modes.
We simulate $R = 500$ pairs of chains and check if they meet within $100$ iterations.
Denoting the number of chains which meet as $N_\tau$, 
we report $i_\tau = N_\tau / R$ as a measure of efficiency in meeting.
Regarding the choice of $\stepsize{}, L$, it is known that HMC is sensitive to the total trajectory length $\stepsize{} L$ in multi-modal distributions: it requires the Hamiltonian simulation long enough to allow jumps between modes.
Therefore, starting with $(\stepsize{}, L) = (0.1, 10)$, we consider two ways of increasing $\stepsize{} L$: sweeping $\stepsize{} \in \{0.1, 0.15, \dots, 0.3\}$ and sweeping $L \in \{10, 15, \dots, 30\}$,
equivalently providing a range of total lengths between $1$ and $3$.
While both means increase the trajectory length, the first approach doesn't introduce additional computation but might lead to larger simulation errors, 
which may then affect the overall performance.
Figure~\ref{fig:gmm} provides $i_\tau$ under such changes of total trajectory lengths for all methods.
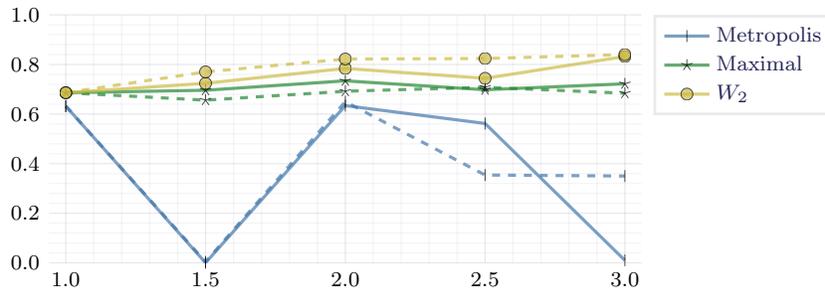
\begin{figure*}[t]
    \centering
    \begin{tikzpicture}[/tikz/background rectangle/.style={fill={rgb,1:red,1.0;green,1.0;blue,1.0}, draw opacity={1.0}}, show background rectangle]
    \begin{axis}[point meta max={nan}, point meta min={nan}, legend cell align={left}, title={}, title style={at={{(0.5,1)}}, anchor={south}, font={{\fontsize{14 pt}{18.2 pt}\selectfont}}, color={rgb,1:red,0.0;green,0.0;blue,0.0}, draw opacity={1.0}, rotate={0.0}}, legend style={color={rgb,1:red,0.1333;green,0.1333;blue,0.3333}, draw opacity={0.1}, line width={1}, solid, fill={rgb,1:red,1.0;green,1.0;blue,1.0}, fill opacity={0.9}, text opacity={1.0}, font={{\fontsize{8 pt}{10.4 pt}\selectfont}}, at={(1.02, 1)}, anchor={north west}}, axis background/.style={fill={rgb,1:red,1.0;green,1.0;blue,1.0}, opacity={1.0}}, anchor={north west}, xshift={1.0mm}, yshift={-1.0mm}, width={94.6mm}, height={48.8mm}, scaled x ticks={false}, xlabel={}, x tick style={draw={none}}, x tick label style={color={rgb,1:red,0.0;green,0.0;blue,0.0}, opacity={1.0}, rotate={0}}, xlabel style={at={(ticklabel cs:0.5)}, anchor=near ticklabel, font={{\fontsize{11 pt}{14.3 pt}\selectfont}}, color={rgb,1:red,0.0;green,0.0;blue,0.0}, draw opacity={1.0}, rotate={0.0}}, xmajorgrids={true}, xmin={0.94}, xmax={3.06}, xtick={{1.0,1.5,2.0,2.5,3.0}}, xticklabels={{$1.0$,$1.5$,$2.0$,$2.5$,$3.0$}}, xtick align={inside}, xticklabel style={font={{\fontsize{8 pt}{10.4 pt}\selectfont}}, color={rgb,1:red,0.0;green,0.0;blue,0.0}, draw opacity={1.0}, rotate={0.0}}, x grid style={color={rgb,1:red,0.1333;green,0.1333;blue,0.3333}, draw opacity={0.1}, line width={0.5}, solid}, extra x ticks={{1.1,1.2,1.3,1.4,1.6,1.7,1.8,1.9,2.1,2.2,2.3,2.4,2.6,2.7,2.8,2.9}}, extra x tick labels={}, extra x tick style={grid={major}, x grid style={color={rgb,1:red,0.1333;green,0.1333;blue,0.3333}, draw opacity={0.05}, line width={0.5}, solid}, major tick length={0}}, x axis line style={{draw opacity = 0}}, scaled y ticks={false}, ylabel={}, y tick style={draw={none}}, y tick label style={color={rgb,1:red,0.0;green,0.0;blue,0.0}, opacity={1.0}, rotate={0}}, ylabel style={at={(ticklabel cs:0.5)}, anchor=near ticklabel, font={{\fontsize{11 pt}{14.3 pt}\selectfont}}, color={rgb,1:red,0.0;green,0.0;blue,0.0}, draw opacity={1.0}, rotate={0.0}}, ymajorgrids={true}, ymin={0}, ymax={1}, ytick={{0.0,0.2,0.4,0.6000000000000001,0.8,1.0}}, yticklabels={{$0.0$,$0.2$,$0.4$,$0.6$,$0.8$,$1.0$}}, ytick align={inside}, yticklabel style={font={{\fontsize{8 pt}{10.4 pt}\selectfont}}, color={rgb,1:red,0.0;green,0.0;blue,0.0}, draw opacity={1.0}, rotate={0.0}}, y grid style={color={rgb,1:red,0.1333;green,0.1333;blue,0.3333}, draw opacity={0.1}, line width={0.5}, solid}, extra y ticks={{0.04,0.08,0.12,0.16,0.24000000000000002,0.28,0.32,0.36,0.44000000000000006,0.4800000000000001,0.5200000000000001,0.56,0.6400000000000001,0.6800000000000002,0.7200000000000001,0.7600000000000001,0.8400000000000001,0.8800000000000001,0.92,0.9600000000000001}}, extra y tick labels={}, extra y tick style={grid={major}, y grid style={color={rgb,1:red,0.1333;green,0.1333;blue,0.3333}, draw opacity={0.05}, line width={0.5}, solid}, major tick length={0}}, y axis line style={{draw opacity = 0}}]
    \addplot[color={rgb,1:red,0.2667;green,0.4667;blue,0.6667}, name path={59fed147-545f-4334-966a-bc1dc3912a78}, draw opacity={0.7}, line width={1.2}, solid, mark={|}, mark size={2.25 pt}, mark options={color={rgb,1:red,0.0;green,0.0;blue,0.0}, draw opacity={0.7}, fill={rgb,1:red,0.2667;green,0.4667;blue,0.6667}, fill opacity={0.7}, line width={0.0}, rotate={0}, solid}]
        coordinates {
            (1.0,0.632)
            (1.5,0.0)
            (2.0,0.634)
            (2.5,0.562)
            (3.0,0.01)
        }
        ;
    \addlegendentry {Metropolis}
    \addplot[color={rgb,1:red,0.2667;green,0.4667;blue,0.6667}, name path={7c03378d-8d72-428d-ba15-ec3a3eced6f5}, draw opacity={0.7}, line width={1.2}, dashed, mark={|}, mark size={2.25 pt}, mark options={color={rgb,1:red,0.0;green,0.0;blue,0.0}, draw opacity={0.7}, fill={rgb,1:red,0.2667;green,0.4667;blue,0.6667}, fill opacity={0.7}, line width={0.0}, rotate={0}, solid}, forget plot]
        coordinates {
            (1.0,0.632)
            (1.5,0.004)
            (2.0,0.65)
            (2.5,0.354)
            (3.0,0.35)
        }
        ;
    \addplot[color={rgb,1:red,0.1333;green,0.5333;blue,0.2}, name path={6b83289c-c2f1-494f-bd85-5601703952a9}, draw opacity={0.7}, line width={1.2}, solid, mark={star}, mark size={2.25 pt}, mark options={color={rgb,1:red,0.0;green,0.0;blue,0.0}, draw opacity={0.7}, fill={rgb,1:red,0.1333;green,0.5333;blue,0.2}, fill opacity={0.7}, line width={0.0}, rotate={0}, solid}]
        coordinates {
            (1.0,0.686)
            (1.5,0.696)
            (2.0,0.734)
            (2.5,0.698)
            (3.0,0.722)
        }
        ;
    \addlegendentry {Maximal}
    \addplot[color={rgb,1:red,0.1333;green,0.5333;blue,0.2}, name path={0de0b3c3-9274-4b62-ac67-7c9eabbf8c5c}, draw opacity={0.7}, line width={1.2}, dashed, mark={star}, mark size={2.25 pt}, mark options={color={rgb,1:red,0.0;green,0.0;blue,0.0}, draw opacity={0.7}, fill={rgb,1:red,0.1333;green,0.5333;blue,0.2}, fill opacity={0.7}, line width={0.0}, rotate={0}, solid}, forget plot]
        coordinates {
            (1.0,0.686)
            (1.5,0.656)
            (2.0,0.692)
            (2.5,0.708)
            (3.0,0.684)
        }
        ;
    \addplot[color={rgb,1:red,0.8;green,0.7333;blue,0.2667}, name path={9af03e12-4443-4ed4-9d6f-cb6736eae1c8}, draw opacity={0.7}, line width={1.2}, solid, mark={*}, mark size={2.25 pt}, mark options={color={rgb,1:red,0.0;green,0.0;blue,0.0}, draw opacity={0.7}, fill={rgb,1:red,0.8;green,0.7333;blue,0.2667}, fill opacity={0.7}, line width={0.0}, rotate={0}, solid}]
        coordinates {
            (1.0,0.686)
            (1.5,0.724)
            (2.0,0.784)
            (2.5,0.744)
            (3.0,0.832)
        }
        ;
    \addlegendentry {$W_2$}
    \addplot[color={rgb,1:red,0.8;green,0.7333;blue,0.2667}, name path={00366e4f-ded5-47dd-a3d9-cdb646429f8d}, draw opacity={0.7}, line width={1.2}, dashed, mark={*}, mark size={2.25 pt}, mark options={color={rgb,1:red,0.0;green,0.0;blue,0.0}, draw opacity={0.7}, fill={rgb,1:red,0.8;green,0.7333;blue,0.2667}, fill opacity={0.7}, line width={0.0}, rotate={0}, solid}, forget plot]
        coordinates {
            (1.0,0.686)
            (1.5,0.77)
            (2.0,0.822)
            (2.5,0.824)
            (3.0,0.84)
        }
        ;
\end{axis}

    \end{tikzpicture}
    
    \caption{Meeting efficiency on the mixture of Gaussians target with the total trajectory length $\stepsize{} L$ increasing. Solid lines are from increasing $\stepsize{}$ and dashed ones from increasing $L$.}
    \label{fig:gmm}
\end{figure*}
First, by increasing $\stepsize{} L$,
our proposed methods overall improve the meeting efficiency,
which is not the case for coupled Metropolis HMC.
This can be explained by the following:
for coupled Metropolis HMC, meetings can only happen if two chains are proposed to the same mode.
However, for coupled multinomial HMC,
as long as the trajectories explore common modes,
there is a chance for meeting.
Especially with $W_2$-coupling, 
this chance is further increased by utlizing the actual distances between pairs to find coupling, making it the best in the figure.
Second, regarding the two ways of increasing $\stepsize{} L$, for our proposed methods, 
increasing $L$ appears to be better as we expected.
That said, the gap is relatively small -- coupled multinomial HMC tends to be robust against large $\stepsize{}$, 
which is practically useful as it allows the use of a smaller amount of computation comparing to increasing $L$.
Note that we do not claim or indicate our methods improve the mixing in multi-modal distributions, which by itself is an important and unsolved issue for HMC.
\begin{table*}[t]
    \centering
    \begin{tabularx}{0.52\columnwidth}{l|ccc}
    \toprule
    Momentum & Metropolis & Maximal & $W_2$ \\ \midrule
    Shared & $136.6 \pm 95.8$ &  $112.4 \pm 74.9$ & $103.8 \pm 76.5$ \\ \hline
    Contractive & $\mathbf{39.7} \pm \mathbf{18.9}$ &  $\mathbf{81.3} \pm \mathbf{56.3}$ & $\mathbf{77.2} \pm \mathbf{48.1}$  \\ \bottomrule
    \end{tabularx}
    \caption{Effect of different momentum coupling methods on meeting time for the Banana target.}
    \label{tab:banana}
\end{table*}
Second, to examine the proposed methods on highly non-convex distributions, 
we consider a banana-shaped distribution on $\mathbb{R}^2$,
of which the potential is given by the Rosenbrock function $U(x_1, x_2) = (1 - x_1)^2 + 10 (x_2 - x_1^2)^2\;(x_1, x_2 \in \mathbb{R})$.
As it is done in \citep{heng_unbiased_2019},
we also take this chance to study the effect of other methods for coupling the initial momentums rather than simply sharing them.
Specifically, we consider the contractive coupling from \citep{bou-rabee_coupling_2019}, 
in which the initial momentums $P^1, P^2$ are sampled based on the current positions $Q^1, Q^2$ as follow
\begin{equation*}
\begin{aligned}
    P^1 &\sim \Normal(0, I), \\ 
    P^2 &= \begin{cases}
    P^1 + \kappa \Delta &\text{with prob. } \frac{\Normal \left(\bar{\Delta}^\top P^1 + \kappa |\Delta|; 0, 1 \right)}{\Normal \left(\bar{\Delta}^\top P^1; 0, 1 \right)} \\ 
    P^1 - 2(\bar{\Delta}^\top P^1)\bar{\Delta} &\text{otherwise}
    \end{cases}
\end{aligned}
\end{equation*}
where $\kappa > 0$ is a tuning parameter, $\Delta = Q^1 - Q^2$ is the difference in position space and $\bar{\Delta}$ is the corresponding normalised difference.
With initial states sampled from 
$\Unif([0,1]^2)$%
, we simulated $R = 500$ pairs of coupled chains with $(\stepsize{}, L) = (1/50, 50)$
for maximally 500 iterations with two momentum coupling methods: shared momentum and contractive coupling with $\kappa = 1$.
We summarise means and standard deviations of $\tau$ from $R$ runs in Table~\ref{tab:banana}.

First of all, 
all method with two momentum coupling methods can meet within 150 iterations in such high non-convex setup,
except approximate $W_2$-coupling with contractive momentum.
Also, it can be seen that our methods can also benefit from contractive coupling,
even though it is derived as a maximal coupling \citep{thorisson_coupling_2000} for Metropolis HMC.
This is the reason why coupled Metropolis HMC is largely improved by it.
That is to say,
contractive coupling is an orthogonal method of ours rather than a replacement.
Note that the table should not be used to compare coupled multinomial HMC against coupled Metropolis HMC in terms of meeting time because they have different optimal parameters for meeting in this target.

\end{document}